\documentclass[a4paper,onecolumn,11pt,accepted=2021-05-25]{quantumarticle}
\pdfoutput=1
\usepackage[utf8]{inputenc}
\usepackage[english]{babel}
\usepackage[T1]{fontenc}
\usepackage{amsmath}

\usepackage[numbers,sort&compress]{natbib}

\usepackage{tikz}
\usepackage{lipsum}

\usepackage{amsthm}
\newtheorem{theorem}{Theorem}
\usepackage{xcolor}
\usepackage{bm}
\usepackage{bbold}
\usepackage{amsmath}
\usepackage{amssymb}
\usepackage{amsthm}
\usepackage{epsfig}
\usepackage{verbatim}
\usepackage[T1]{fontenc}
\usepackage{array}

\newcommand{\ket}[1]{\left|#1\right\rangle}
\newcommand{\bra}[1]{\left\langle #1\right|}
\newcommand{\bracket}[1]{\left\langle #1\right\rangle}

\newcommand{\norm}[1]{\left\lVert#1\right\rVert}

\usepackage[colorlinks=true,bookmarks=false,citecolor=blue,linkcolor=blue,hyperfootnotes=true,urlcolor=blue]{hyperref}

\begin{document}

\title{Efficient qubit phase estimation using adaptive measurements}

\author{Marco A. Rodr\'iguez-Garc\'ia}
\email{marg@ciencias.unam.mx}
\affiliation{Instituto de Investigaciones en Matemáticas Aplicadas y en Sistemas, Universidad Nacional Autónoma de México, Ciudad Universitaria, Ciudad de México 04510, Mexico}

\author{Isaac P\'erez Castillo}
\email{iperez@izt.uam.mx}
\affiliation{Departamento de F\'isica, Universidad Autónoma Metropolitana-Iztapalapa, San Rafael Atlixco 186, Ciudad de México 09340, Mexico}

\author{P. Barberis-Blostein}
\email{pbb@iimas.unam.mx}
\affiliation{Instituto de Investigaciones en Matemáticas Aplicadas y en Sistemas, Universidad Nacional Autónoma de México, Ciudad Universitaria, Ciudad de México 04510, Mexico}
\maketitle

\begin{abstract}
 Estimating correctly the quantum phase of a physical system is a
  central problem in quantum parameter estimation theory due to its
  wide range of applications from quantum metrology to cryptography.
  Ideally, the optimal quantum estimator is
  given by the so-called quantum Cram\'er-Rao bound, so any
  measurement strategy aims to obtain estimations as close as possible
  to it. However, more often than not, the current state-of-the-art
  methods to estimate quantum phases fail to reach this bound as they
  rely on maximum likelihood estimators of non-identifiable likelihood
  functions. In this work we thoroughly review various schemes for
  estimating the phase of a qubit, identifying the underlying problem
  which prohibits these methods to reach the quantum Cram\'er-Rao
  bound, and propose a new adaptive scheme based on covariant
  measurements to circumvent this problem. Our findings are carefully
  checked by Monte Carlo simulations, showing that the method we
  propose is both mathematically and experimentally more realistic and
  more efficient than the methods currently available.
  \end{abstract}

\section{Introduction}
The aim of statistical estimation theory in classical systems is to
estimate a probability distribution based on a series of observations.
More precisely, given a statistical model
$\mathcal{P} =\left\{ p(x \mid \theta) \mid \theta \in \Theta\subset
  \mathbb{R}^n \right\}$ and a series of observations
$\{X_1,\ldots, X_N\}$ generated by a particular distribution density
$p(x \mid \theta)$, the aim is to design an estimator
$\hat{\theta}(X_1,,\ldots, X_N)$ which is as close as possible to the
parameter $\theta$, measured in terms of a cost function, usually the
Mean Square Error (MSE), between $\hat{\theta}$ and $\theta$. A fairly
remarkable result, the celebrated Cram\'er–Rao bound (CRB), states
that the variance of an unbiased estimator, is at least as high as the
inverse of the Fisher information of $p(x \mid \theta)$. The
estimators that saturate this bound are called efficient. An example of an 
efficient estimator is the maximum likelihood, that saturates the
CRB when the number of measurements tends to infinity.

This classical result has, however, some subtleties, namely: for the
maximum likelihood estimator to be consistent (i.e. the estimator
converges to the parameter when the number of outcomes tends to
infinity), the likelihood function must be identifiable (i.e. the
likelihood has a unique global maximum), $\Theta$ must be a compact
set, and all probability distributions in the statistical model must
have the same support. Statistical models satisfying the previous
conditions are called regular models. 

While the actual role that the process of measurement plays in
extracting information from a classicaly system is not particularly
relevant, the situation is completely different in a quantum system,
making the process of parameter estimation mathematically more
involved. Suppose $\rho$ is an unknown density operator acting on a
Hilbert space $\mathcal{H}$, that represents the state of the system.
Let us assume that $\rho$ belongs to a particular subset of
parametrized states
\begin{equation}
  \label{QSM}
  \mathcal{S} = \left\{ \rho_{\theta}; \theta  \in \Theta \subset \mathbb{R}^{n} \right\}\,.
\end{equation}
The parametrized states $\rho_{\theta}$ are the result of the
evolution of an initial probe state $\rho$ by a trace-preserving
dynamical process dependent upon a parameter $\theta$. Note that
different initial conditions give rise to different subsets of
parametrized states.

Quantum parameter estimation aims to produce the best estimation for $\theta$ under the premise that the system is prepared in a state of the parametric family $\mathcal{S}$. Notice that, unlike its classical counterpart, quantum estimation consists of two parts:  the choice of the measurements to perform on the system, and then the processing of their outcomes through an estimator $\hat{\theta}$. Note that measured data is a set of random outcomes with a probability distribution depending on the parameter, and an estimator is a random variable whose outcomes estimates the parameter.  Thus, the pair measurement-estimator forms a strategy of estimation in the quantum case one can play around with seeking for an optimal estimation strategy. Indeed, each choice of measurement defines a classical statistical model, whose lowest possible CRB maximises the Fisher information. Thus the maximum of all possible Fisher informations, over the space of measurements allowed by the postulates of Quantum Mechanics, is called the Quantum Fisher Information (QFI), and its inverse, the Quantum Cram\'er-Rao Bound (QCRB). In other words, the QCRB is the minimum over all possible MSE of any possible estimation strategy allowed by Quantum Mechanics. If the Fisher information for a particular choice of measurement $M$ is equal to the QFI, this measurement is called optimal.\\
One way to find an optimal measurement is to use the set of  operators $\left\{ M(j) \right\}$  to represent the measurement, where each $M(j)$ corresponds to  a projector onto each of the  eigenspaces of the symmetric logarithmic derivative (SLD) $\mathcal{L}(\theta)$ \cite{Braunstein1994},  defined by
\begin{equation}
  \label{eq:SLD}
  \frac{\partial \rho_{\theta}}{\partial \theta} = \frac{1}{2}\left(\mathcal{L}(\theta)\rho_{\theta} + \rho_{\theta}\mathcal{L}(\theta) \right)\,.
\end{equation}

It turns out that, in general, this measurement is so-called locally
optimal. This means that the optimal measurement depends on the actual
value of the parameter we want to estimate and the classical Fisher
information given by that measurement is a local maximum in the parameter
space that characterizes the measurement. Thus estimating the parameter
using a locally optimal measurement is rather impractical. Two
approaches have been developed to overcome this problem. The first one
is based on adaptive estimation schemes which updates a guess for a
locally optimal measurement \cite{Fujiwara2011,
  Barndorff-Nielsen2000}. The second method seeks a set of initial
conditions that do not depend on the unknown parameter
\cite{Chapeau2016, Toscano2017}. Both methods are based on the maximum
likelihood estimator (MLE) which implies that in order to obtain it, the aforementioned set of regularity conditions must be met
\cite{Robert, Lehmann1998,CaseBerg:01}. We will see that the MLE
derived from optimal measurements may fail to satisfy these conditions
and thus these two approaches fail to attain the QCRB.

The problem of non-identifiable likelihood functions often appears  when estimating the quantum phase of a system, a particular problem with a wide range of applications from quantum metrology to cryptography, as many problems can be recast into a quantum phase estimation \cite{Jon,  PhysRevA.73.033821, Rock}. In this framework, when trying to use locally optimal measurements \cite{Giovann2004, Frowis2014, Pezze2014, Huang2017,Toscano2017}, it may happen that adaptive schemes  fail to reach the QCRB \cite{Fuji2} due to the fact that regularity conditions are not met \cite{Fujiwara2011}. A paradigmatic example of non-identifiablity is provided when trying to perform qubit phase estimation which are, in general, unable to reach the QCRB \cite{Chapeau2016,Barndorff-Nielsen2000}. \\
The main goal of the present work is to introduce a new adaptive
scheme for quantum phase estimation which saturates the QCRB
independently of the initial condition that generates $\mathcal{S}$. We will first show that the Adaptive Quantum State Estimation method (AQSE), a general method for parameter estimation \cite{Yamagata2013AsymptoticQS}, does not achieve the QCRB when applied to a qubit phase estimation.  The reason why AQSE does not converges is that the statistical model, built with the measurement that maximizes the Fisher information (the locally optimal measurement), is not identifiable. We will further show that the non-identifiability problem is solved by taking a sample of measures from a covariant measurement. The covariant measurement is identifiable and is chosen to minimize the MSE for one measure. With the covariant sample, we can then construct a confidence interval where the underlying statistical model is now regular. Then we will apply the AQSE method inside the confidence interval.\\
The paper is organized as follows. We start by giving a brief review
of quantum estimation theory with special emphasis in covariant
estimation techniques for periodic parameter estimation
(section~\ref{lb:quantum_parameter_estimation}). Then, we review
different estimation strategies for the phase of a qubit and discuss
their weakness and strengths, in particular we discuss the effects of
non-identifiable probabilistic models in the estimation error
(section~\ref{lb:optimal_qubit_phase_estimation}). Once we understand
the problems with different estimation strategies, we present a two-step estimation scheme. Numerical simulations suggest that this scheme
is able to reach the QCRB (section
\ref{lb:arbitrary_initial_condition}). We end the paper with a summary
of the main results and our conclusions (section
\ref{summary_conclusions}).

\section{Quantum parameter estimation}
\label{lb:quantum_parameter_estimation}
Recall that, given a set of independent observations $\{x_1,\ldots, x_N\}$ from
a random variable $X$ with probability density $p(x \mid \theta)$, the
likelihood function is defined as
\begin{equation}
  \label{eq:likelihood_def}
  L(\theta \mid x_1,...,x_N) = \prod_{i=1}^{N} p(x_i \mid \theta).
\end{equation}
From here, we can derive the maximum likelihood estimator for $\theta$, whose
mean square error obeys the CRB \cite{Lehmann1998} 
\begin{equation}
  \label{eq:CRB}
  \text{MSE}(\hat{\theta}) \geq \frac{1}{N F(\theta; X)}\, ,
\end{equation}  
where the MSE is defined by
\begin{equation}
  \label{eq:mse}
\text{MSE}(\hat{\theta}) =  E_{\theta}\left[ \left(\hat{\theta} - \theta \right)^2 \right]\,,
\end{equation}
and $F(\theta; X)$ is the Fisher information:
\begin{equation}
  \label{eq:fi}
  \begin{split}
    F(\theta; X) = E_{\theta} \left[\left(\frac{\partial }{\partial
          \theta}\log p(x \mid \theta)\right)^2\right].
  \end{split}  
\end{equation}
Intuitively, the Fisher information  quantifies how much information carries a
sample about the unknown parameter.  The ultimate aim in classical parameter
estimation is to find the estimator that achieves the Cramér-Rao bound. If the
statistical model $\left\{ p(x \mid \theta) \mid \theta \in \Theta \right\}$ for a random variable $X$ is regular the maximum likelihood estimator produced by a sequence of $N$ independent and identically distributed of $X$ can attain the Cram\'er-Rao bound   asymptotically for large $N$ \cite{Lehmann1998, Robert}.\\
In the quantum case, we  must first describe the process of measuring the system. This is better achieved by using the concept   of Positive Operator-Valued Measures (POVMs).  A POVM with outcomes on a set $\mathcal{X}$ is a family of bounded positive operators
\begin{equation*}
M= \left\{M(B); B \text{ is a Borel set in } \mathcal{X} \right\}
\end{equation*}
acting over the Hilbert space of the system $\mathcal{H}$ such that $M(\mathcal{X}) = I$. When the state of the system is given by the density operator $\rho_{\theta}$, a POVM $M$ specifies the conditional probability of obtaining the event $B$ by  Born's rule
\begin{equation}
  \label{eq:br}
  P\left(B \mid \theta  \right) = \text{Tr}\left[M(B) \rho_{\theta} \right].
\end{equation}  
Thus, in quantum parameter estimation, the Fisher information, given by Eq. (\ref{eq:fi}), is a function of the POVMs and, henceforth, we will denote it  as  $F(\theta; M)$. This implies that, according to Eq. \eqref{eq:CRB}, maximising $F(\theta; M)$ over the set of POVMs gives the lowest CRB. This results into the so-called   Quantum Fisher Information (QFI), which we will denote as  $F_Q(\theta)$, and its lowest bound is so-called the Quantum Cram\'er-Rao Bound (QCRB) \cite{Helstrom1969, Holevo1973, Holevo2011}.\\
For the particular case of a scalar parameter estimation, the QFI has the following form \cite{Helstrom1969, Holevo2011, Braunstein1994}
\begin{equation}
\label{Fq}  
  F_Q(\theta) = \text{Tr}\left[\rho_{\theta}\mathcal{L}(\theta)^2 \right],
\end{equation}  
where  $\mathcal{L}(\theta)$  is the symmetric logarithmic derivative, also called quantum score, defined by Eq. \eqref{eq:SLD}. When the system is in a pure state $\rho_\theta = \ket{\psi_{\theta}}\bra{\psi_{\theta}} $,  the quantum score is easy to calculate \cite{Paris2009}, obtaining $\mathcal{L}(\theta)=2\frac{\partial \rho_\theta}{\partial \theta}$, so that the QFI becomes:
\begin{equation}
  F_{Q}(\theta)=4 ~\text{Tr} \left[ \left( \frac{\partial }{\partial \theta} \left( \ket{\psi_{\theta}}\bra{\psi_{\theta}} \right)   \right)^2  \ket{\psi_{\theta}}\bra{\psi_{\theta}}    \right].
    \label{eq:qfi_pure}
  \end{equation}
 Notice that the set of  POVMs  $\left\{M_{\theta}(j) \right\}$
 constructed as the projections of the quantum score $\mathcal{L}(\theta)$ do
 depend on the parameter $\theta$ one is trying to infer
 \cite{Braunstein1994}. A natural way around this is to introduce
 adaptive schemes to estimate $\theta$. One  approach relies on
 adaptive quantum estimation schemes based  on locally optimal POVMs
 that could, in principle, asymptotically  construct the optimal POVM
 without knowing the parameter beforehand  \cite{Nagaoka,
   Fujiwara2011, Berry2001, Huang2017}. Nevertheless, a  set of
 precise mild regularity conditions for each statistical model
 involved in the method are required. For instance, in the adaptive
 quantum state estimation (AQSE) method \cite{Nagaoka, Fujiwara2011},
 it is necessary to assume regular statistical models for every
 measurement to guarantee a consistent and efficient estimator.
 However, for estimation problems on which the quantum parameter is
 periodic, as it is the case for phase estimation, the  likelihood
 functions produced by locally optimal POVMs are not identifiable
 \cite{Chapeau2016,Berry2009}. Consequently, there is no mathematical
 reason that ensures  the saturation of the CRB. 
The second approach searches specific initial conditions, for which the optimal  POVM does not depend on the unknown value of $\theta$ \cite{Toscano2017,    Chapeau2016}.  Hence, in principle,  performing an extensive independent  sequence of this POVM, along with the maximum likelihood estimator, it is  possible to achieve the quantum Cramér-Rao bound. Nonetheless, for periodic  parameter estimation, this method  often produces non-identifiable  likelihood functions, and  as a consequence, the CRB is not attained.\\
We will show now that, for  the case  of quantum  phase estimation, we can solve the problem of non-identifiable  likelihood functions by using the formalism of covariant estimation, which  considers the symmetries of the system and can saturate the  QCRB under specific initial conditions \cite{Holevo2011}.

\subsection{Covariant estimation}
The following discussion is based on \cite{Holevo2011}.
In quantum parameter estimation, each $\rho_{\theta} \in \mathcal{S}$ obtains its parametric dependency through a physical transformation $U_\theta$. Usually, the set of $\left\{ U_{\theta} \right\}_{\theta \in \Theta}$ forms a group, which induces an action over $\mathcal{S}$.  When the family $\mathcal{S}$ is invariant under the conjugation of $U_{\theta}$, for any $\theta \in \Theta$, we say that the problem of estimation involves a symmetry. The problem of quantum parameter estimation  involving symmetries is called covariant. In the following,  we summarise the approach of covariant estimation.\\
Let $G$ be a locally compact Lie group acting transitively over the parametric set $\Theta$. Thus for any base point $\theta_0 \in \Theta$, we have
\begin{equation}
  \Theta = \theta_0 \cdot G \cong G/K\,, 
\end{equation}
where $K$ is the stabilizer subgroup of $\theta_0$. According to Wigner's theorem \cite{Wigner}, the Hilbert space $\mathcal{H}$ of the system  has a unitary $G$-representation $U=\left\{U_g\right\}_{g \in G}$. Hence, the probe states can be transformed by the conjugation
\begin{equation}
  \label{eq:rho_t}
  \rho \mapsto U_{g} \rho U_{g}^{\dagger}\,,
\end{equation}
where $U_g$ is an element of the $G$-unitary representation of $\mathcal{H}$. Specifically, we say that we have a quantum covariant estimation problem whenever
\begin{equation}
  \label{eq:cep}
\left\{ U_{g} \rho_{\theta} U_{g}^{\dagger} \mid  \theta \in \Theta \right\} = \mathcal{S} \quad \forall g \in G\,.
\end{equation}
In this framework, a POVM $M$ on the system $\mathcal{H}$ taking values in $\Theta$ is called \emph{covariant} with respect the $G$-unitary representation $\left\{U_g \right\}$, if for every event $B$ and $g \in G$, it satisfies that
\begin{equation}
  \label{eq:covPOVM_def}
  M(Bg^{-1}) = U_g^{\dagger} M(B) U_g,
\end{equation}
where $Bg^{-1} = \left\{\theta \mid \theta = g \theta', \theta' \in B \right\}$. Note that this condition is equivalent to
\begin{equation}
  \label{eq:probcov_def}
  P\left(B \mid \rho \right) = P\left(gB \mid U_{g} \rho U_{g}^{\dagger}\right)\,.
\end{equation}
We will denote as $\mathcal{M}\left(\Theta\right)$ the set of POVMs with outcomes in $\Theta$. Similarly, we will denote as as $\mathcal{M}\left(\Theta, U \right)$ the set of the covariant POVMs with respect the $G$-unitary representation.\\
The class of covariant POVMs takes advantage of group symmetries. Specifically, when the outcomes of the measurement are considered as the estimates, and the cost function $c(\hat{\theta},\theta) $ under consideration is $G$-invariant, that is $c(\hat{\theta},\theta) = c(g\hat{\theta}, \theta)~ \forall g \in G$, one has that
\begin{equation}
  \label{eq:costinv}
 \begin{split} 
   E_{g\theta, M}\left[ c(\hat{\theta}, g\theta) \right] &=  \int_{\Theta} c(\hat{\theta}, g \theta) Tr\left[ M(d\hat{\theta}) \rho_{g \theta} \right] =  \int_{\Theta} c(g \hat{\theta'}, g \theta) Tr\left[ M(d\hat{\theta'}) \rho_{\theta} \right] \\
   &= E_{\theta,M}\left[ c(\hat{\theta}, \theta) \right], \, \forall \, M \in  \mathcal{M}\left(\Theta, U\right) \text{ and } \hat{\theta}' := g^{-1}\hat{\theta}.
  \end{split} 
\end{equation}
As a result, we can find optimal measurements on $\mathcal{M}\left(\Theta, U\right)$ in the Bayesian approach, by taking the average of $c(\theta,\hat{\theta})$ over a prior probability measure. The following theorem ensures this assertion.

\begin{theorem}[\cite{Holevo1973}]
  \label{t:qhs}
  When $\Theta$ is compact,
  \begin{equation}
    \label{eq:hunt}
    \begin{split}
      \min_{M \in \mathcal{M}\left(\Theta \right)} \int_{\Theta} E_{\theta,M}\left[ c(\hat{\theta}, \theta) \right] \mu(d \theta)
       = \min_{M \in \mathcal{M}\left(\Theta, U \right)} \int_{\Theta} E_{\theta,M}\left[ c(\hat{\theta}, \theta) \right] \mu_{\Theta}(d \theta).
    \end{split}
  \end{equation}
where $\mu_{\Theta}$ is the Haar measure over $\Theta$.  
\end{theorem}
The previous result restricts the search for optimal measurements to the class $\mathcal{M}\left(\Theta, U \right)$. This fact is particular useful as, according to the following theorem, the covariant POVMs can be characterised as follows.
\begin{theorem}[\cite{Holevo2011}]
  Let $P_0$ be a Hermitian positive operator on $\mathcal{H}$, commuting with  the operators $\left\{U_g\right\}_{g \in K}$ and satisfying
  \begin{equation}
    \label{eq:pc}
    \int_\Theta  U_{g(\theta)} P_0 U_{g(\theta)}^{\dagger} \mu_{\Theta}(d \theta) = I,
  \end{equation}
  where $g(\theta) \in G$ is any representative element of the equivalence class  $\theta \in \Theta$. Then, a POVM $M(d \theta)$ with outcomes in $\Theta \cong G/K$ is covariant if and only if has the form
  \begin{equation}
    \label{eq:covchar}
    M(d \theta) = U_{g(\theta)} P_0 U_{g(\theta)}^{\dagger} \mu_{\Theta}(d \theta).
  \end{equation}
\end{theorem}

A particular case, which is relevant for quantum phase estimation, is when the stability group is the trivial one, that is when $K=\{e\}$ with $e$ the identity element. Here, we have that $\Theta \cong G$, that is, the parametric space $\Theta$ is on one-to-one correspondence to the elements of the group $G$. This implies, in turn, that $E_{\theta,M}\left[ c(\hat{\theta}, \theta) \right]$ is constant for all $\theta \in \Theta$ and,  as consequence of theorem (\ref{t:qhs}), the optimization of the expectation of the invariant cost function can be restricted to the set of covariant POVMs. Notice that for quantum phase estimation, there is a little caveat: here the  one-parameter symmetry group $G$ isomorphic to $\Theta = [0, 2 \pi )$ but this is  not a compact set. In practice, we can consider it to be compact, assuming that the unknown parameter is an interior point of $\Theta$. Thus, we can analyze the CRB for any covariant POVM.\\
 We proceed to discuss in more detail how these ideas apply to quantum  phase
 estimation for qubit states.\\
 
\section{Phase estimation strategies in a qubit}
\label{lb:optimal_qubit_phase_estimation}
Let $\mathcal{S} = \left\{ \rho_{\theta}; \, \theta \in \Theta = [0,2\pi)\right\}$ be a parametric family of density operators, on a $2$-dimensional Hilbert space $\mathcal{H}$. In this case, each state $\rho_{\theta} \in \mathcal{S}$ represents the state of a qubit. Here, the states $\rho_{\theta}$ obtain their parametric dependency applying an arbitrary unitary transformation
$U_{\theta}$ over a probe state $\rho$ on $\mathcal{H}$, that is
\begin{equation}
  \label{eq:qeq}
  \rho_{\theta} = U_{\theta}\rho U_{\theta}^{\dagger}\,.
\end{equation}
An arbitrary unitary transformation (up to a global phase) on a $2$-dimensional Hilbert space can be written using the generators of the Lie algebra su$(2)$ and has the following form \cite{Eugen,Nielsen2000}
\begin{equation}
  \label{eq:utq}
  U_{\theta} = e^{-i \theta \vec{n} \cdot \frac{\vec{\sigma}}{2}}\,,
\end{equation}
where $\vec{n}=n_x\hat{x}+n_y\hat{y}+n_z \hat{z}\in \mathbb{R}^3$ is a unit vector and $ \vec{\sigma} = \sigma_1 \hat{x} + \sigma_2 \hat{y} + \sigma_3 \hat{z}$, with $\sigma_i$,  for $i = 1,2,3$,   the Pauli matrices. The estimation problem in this case consists in finding the best strategy to estimate the phase $\theta \in \Theta$ of the exponential operator appearing in Eq. (\ref{eq:utq}).\\
To analyse the QCRB in this case it is better to work using the Bloch sphere
representation for a qubit state. The probe state can be written as
\begin{equation}
  \label{eq:qubirbs}
 \rho = \frac{1}{2}\left(I+\vec{a}\cdot \vec{\sigma}\right) \,,\quad\quad \norm{\vec{a}} \leq 1\,,
\end{equation}
and the transformation $U_{\theta}$ can be seen as the rotation of the probe Bloch vector $\vec{a}$ by the angle $\theta$ around the axis $\vec{n}$. We will denote the transformed vector as $\vec{a}(\theta)$, so the transformed qubit state $\rho_{\theta}$ is solely determined by $\vec{a}(\theta)$. One can show that, for this problem, the quantum score takes the following form \cite{Chapeau2016}:
\begin{equation}
  \label{eq:qscorequbit}
  \mathcal{L}(\theta) = \left(\vec{n} \times \vec{a}(\theta) \right) \cdot \vec{\sigma}\,,
\end{equation}
where
\begin{equation}
  \label{eq:atheta}
  \vec{a}(\theta) = \cos(\theta)\vec{a} + \left( \sin(\theta) (\vec{n} \times  \vec{a})
  \right) + \left( 2  (\vec{n} \cdot \vec{a})  \sin^2\left(
      \frac{\theta}{2} \right) \right) \vec{n}\, .
\end{equation}
This implies that, according to  Eq. (\ref{Fq}), the quantum Fisher information becomes
\begin{equation}
  \label{eq:qfq}
F_Q = \norm{\vec{a}}^2 - (\vec{a} \cdot \vec{n})^2\,.
\end{equation}
Thus, if the probe qubit is a pure state (i.e., $\norm{\vec{a}}= 1$) and the
Bloch vector $\vec{a}$ is orthogonal to the rotation axis $\vec{n}$, then
Eq.~(\ref{eq:qfq}) achieves their maximal value $F_Q^{\text{max}} = 1$. We will
refer to this as the optimal initial conditions. Let us proceed to discuss the
different estimation strategies for the phase $\theta$.

\subsection{Locally Optimal POVMs}
\label{lb:LOP}
From the expression of the quantum score, given by Eq.~(\ref{eq:qscorequbit}), it is straightforward obtain that its projection operators are given by:
\begin{equation}
  \label{eq:projqs}
  M_{g}(0) = \frac{1}{2}\left( I + \frac{\left(\vec{n} \times \vec{a}(g) \right)}{\norm{ \vec{n} \times \vec{a}(g)  }}  \cdot \vec{\sigma}   \right)\,,\quad\quad  M_{g}(1) = \frac{1}{2}\left( I - \frac{\left(\vec{n} \times \vec{a}(g) \right)}{\norm{ \vec{n} \times \vec{a}(g)  }}  \cdot \vec{\sigma}   \right)\,,\quad\quad  g \in \Theta\,. 
  \end{equation}
 Hence, the family of POVMs is $M_g = \left\{ M_{g}(0), M_{g}(1) \right\}$, with outcomes in the set $\mathcal{X}= \left\{0,1\right\}$,   yield  the following classical Fisher information
\begin{equation}
  \label{eq:fipj}
F(\theta ; g) = \frac{F_Q \cos^2(\theta- g)}{1-F_Q\sin^2(\theta-g)}\,.  
\end{equation}
From Eq. (\ref{eq:fipj})  it is fairly obvious to realize that $F(\theta, g) = F_Q(\theta)$ whenever $g = \theta$ for any $\theta \in \Theta$. That is, the POVM $M_{\theta}=\left\{ M_{\theta}(0), M_\theta(1) \right\}$ is locally optimal. Note, however,  that if $\vec{a} \perp\vec{n}$ then
\begin{equation}
  \label{eq:fipjmax}
F(\theta ; g) = \frac{ \cos^2(\theta- g)}{1-\sin^2(\theta-g)} = 1 = F_{Q}^{\text{max}}.  
\end{equation}
In other words, when the probe state is prepared in the optimal initial
conditions, the POVM $M_g$ is optimal, independent of $\theta \in \Theta$.
Therefore, any POVM $M_g$ is optimal.

At this point, one may think that a POVM $M_g$ produces an \textit{optimal
  estimation strategy} in the sense that if we were to perform a sequence of $N$
independent measurements using the maximum likelihood estimator, it should be
possible to saturate the QCRB in the asymptotic regime. This is, however, false:
the POVM $M$ yields a non-identifiable likelihood function, and thus the maximum
likelihood estimator is not consistent. Indeed, when a sequence of $N$
identically quantum systems with Hilbert space $\mathcal{H}$ is prepared in the
same state $\rho_{\theta}$, the composite system is described by the $N$-tensor
product $\underbrace{\mathcal{H} \otimes\cdots \otimes
  \mathcal{H}}_{N\text{-times}}$ and its state is described by an $N$-fold
tensor product state $\underbrace{\rho_{\theta} \otimes \cdots\otimes
  \rho_{\theta}}_{N-\text{ times}}$. Then, the likelihood function produced by
the outcomes $\vec{x} = \left( x_1,...,x_N \right) \in \mathcal{X}^{N}$ from a
sequence of $N$ independent POVMs $M(g)$ is
\begin{equation}
  \label{eq:likeP}
  \begin{split}
  L(\theta \mid \vec{x}; g) &= \text{Tr}\left[ \left(  M_{g}(x_1) \otimes \cdots \otimes M_{g}(x_N) \right) \rho_{\theta} \otimes \cdots\otimes
    \rho_{\theta} \right] = \prod_{i=1}^{N} \text{Tr} \left[   M_{g}(x_i) \rho_{\theta} \right] = \prod_{i=1}^{N}p(x_i \mid \theta; g ) \\
  &= p(0 \mid \theta; g)^{m}p(1 \mid \theta; g)^{N-m}\,,
 \end{split} 
\end{equation}
where $m$ is the number of $0$'s in $\vec{x}$, and
\begin{equation}
  \label{eq:p_mi}
  p\left( x \mid \theta; g \right) = \begin{cases}
    \text{Tr}\left[ M_g(1)\rho_{\theta} \right]=\frac{1}{2} \left[ 1 + \sin\left( \theta - g \right) \sqrt{F_Q} \right] \,,\quad x=1\,,  \\
    \text{Tr}\left[ M_g(0)\rho_{\theta} \right]=\frac{1}{2} \left[ 1 - \sin\left( \theta - g \right) \sqrt{F_Q} \right]\,, \quad  x=0\,.
  \end{cases}
\end{equation}
By looking for the MLE of $\theta$, denoted $\hat{\theta}_{\text{MLE}}$, we obtain:
\begin{equation}
  \label{eq:mle_P}
  \hat{\theta}_{\text{MLE}} = \text{arg max}_{\theta \in \Theta} L(\theta \mid \vec{x}; g) = \text{arcsin }\left( \frac{1}{\sqrt{F_Q}}\left[ 1-\frac{2m}{N} \right]    \right) + g\,,
\end{equation}
which returns two values in $\Theta = \left[0, 2 \pi \right)$, indicating the
the likelihood is non-identifiable. Thus, the Fisher information alone is not
enough to characterize the error of this estimation strategy. The previous
discussion also can be found in \cite{Chapeau2016}.

Now we consider the case of POVMs with $k\geq 3$ outcomes. Any element
of a POVM $M = \left\{M(i) \right\}$ with outcomes in a set
$\mathcal{X} = \left\{1,2,..., k \right\}$ can be written as
$M(i) = f_0^{(i)}I + \vec{f}^{(i)} \cdot \vec{\sigma}$, where
$\sum_{i=1}^{k} f_0^{(i)} = 1$ and
$\sum_{i=1}^{k} \vec{f}^{(i)} = \vec{0}$. When a qubit is in a state
$\rho_{\theta}$, the probability of measuring the outcome $i$ is
\begin{equation}
  \label{eq:pi}
  p(i \mid \theta) = \text{Tr}\left[\rho_{\theta}M(i)\right] = f_0^{(i)}+ \vec{a}(\theta) \cdot \vec{f}^{(i)}.
\end{equation}

A necessary and sufficient condition for identifiability is that, for any
parameters $\theta_1, \theta_2 \in \Theta$, the set of equations $p(i \mid \theta_1)
= p(i \mid \theta_2)$, $i = 1,...,k$ admits only one solution in the parametric
space $\Theta$ \cite{Marica}.

The problem of identifiability could be avoided by considering POVMs
with several outcomes because it is easier to satisfy the requirements
for identifiability given above. However, a POVM that produces an
identifiable likelihood is not necessarily optimal. Outside the
optimal initial conditions, there is no measurement $M$ that does not
depend on $\theta$ and that saturates the QCRB
\cite{Barndorff-Nielsen2000}. In \cite{Chapeau2016} it is shown a
family of identifiable POVMs with more than $2$ outcomes, however,
except for the optimal initial condition, this family is not locally
optimal. The results in \cite{Barndorff-Nielsen2000, Chapeau2016} do
not forbid the existence of an identifiable and locally optimal POVM
with more than $2$ outcomes, but we could not find one. Note also that
for a small number of measurements, the locally optimal POVM  does not
necessarily saturate the QCRB. Nevertheless it is possible to built a
POVM that, for one measurement, produces an identifiable model that
minimizes the mean square error. We will review now how to build this
POVM by using the covariant approach \cite{Holevo2011}. Once we show
how this POVM is built, we proceed with our proposal to achieve the
QCRB for qubit phase estimation.

\subsection{Covariant phase estimation for qubits}
\label{lb:sub_covariant}
The following discussion is based on \cite{Holevo2005}, where $M_{*}$ is deduced
for any shift parameter.

As we have previously mentioned, phase estimation estimation is covariant with trivial stabilizer group since the group $G = \left[0, 2 \pi \right)$ equipped with addition modulo $2\pi$ acts over itself. Then, $\Theta$ is isomorphic to $G$. Moreover, $U=\left\{U_{\theta} \right\}_{\theta \in \left[0, 2 \pi \right)}$ is a $G$-unitary representation of $\mathcal{H}$. Using Eq. (\ref{eq:covchar}) the covariant POVM $M$  to estimate $\theta$ has the form
\begin{equation}
  \label{eq:povmcovq}
M(d \, \hat{\theta}) =  e^{-i \hat{\theta} \vec{n} \cdot \frac{\vec{\sigma}}{2}} P_o e^{i \hat{\theta} \vec{n} \cdot
  \frac{\vec{\sigma}}{2}} \frac{d \hat{\theta}}{2 \pi}\, ,
\end{equation}
where $P_o$ is a positive operator satisfying Eq. (\ref{eq:pc}).\\
The generator of the unitary representation is the operator $J:=\vec{n} \cdot \frac{\vec{\sigma}}{2}$, which has a spectral decomposition 
\begin{equation}
  \label{eq:sdec}
  \begin{split}
    &J = \frac{1}{2} \ket{\tfrac{1}{2}}\bra{\tfrac{1}{2}} - \frac{1}{2} \ket{-\tfrac{1}{2}} \bra{-\tfrac{1}{2}}\,,
    \end{split}
\end{equation}
where have used that $\ket{\pm \tfrac{1}{2}} \bra{\pm \tfrac{1}{2}} = \frac{1}{2}\left(I \pm \vec{n} \cdot \vec{\sigma} \right)$.
In the basis $\left\{\ket{\tfrac{1}{2}}, \ket{-\tfrac{1}{2}} \right\}$, one can express  $M(d \, \hat{\theta})$ as
\begin{equation}
  \label{eq:mbasis}
  M(d \, \hat{\theta}) = \frac{d \hat{\theta}}{2 \pi} \sum_{m,n \in \left\{-\frac{1}{2}, \frac{1}{2} \right\}}  e^{i \hat{\theta} \left(n-m\right) }p_{nm}\ket{m}\bra{n}, 
\end{equation}
where $p_{m n} =\bracket{m |P_o|n }$. Thus, any covariant POVM is characterized by the real numbers $ 0 \leq p_{n m} \leq 1$. Then, by Eq. (\ref{eq:br})
\begin{equation}
  \label{eq:brspin}
  \begin{split}
    &P(\hat{\theta} \in B | \theta ) = \int_{B} \frac{d \hat{\theta}}{2 \pi}\sum_{m,n \in \left\{-\frac{1}{2}, \frac{1}{2} \right\}} e^{i(n-m)(\hat{\theta}-\theta)}p_{mn}\bra{n} \rho \ket{m}. 
    \end{split}
\end{equation}
Note that the measure (\ref{eq:brspin}) is a $2\pi$-periodic probability distribution, so it is natural to consider the moments of the random variable $e^{i\phi}$ instead. The first moment for a circular distribution $p(\phi)$ is defined as
\begin{equation}
  \label{eq:circmom}
  \begin{split}
    E\left[e^{i \phi} \right] = \int_{0}^{2 \pi}  e^{i \phi}  p(\phi) d \, \phi \,,
  \end{split}
\end{equation}
so from here we can estimate the phase as  $\overline{\phi} = \text{Arg}\left( E\left[e^{i \phi} \right] \right)$.\\
To quantify the correct dispersion of the estimates we use the Holevo variance \cite{Holevo2011}
\begin{equation}
  \label{eq:HV}
  V^{H}_{M}(\hat{\theta}) = \mu^{-2}-1\, ,
\end{equation}
where $M \in \mathcal{M}\left(\Theta, U\right)$ and $\mu = \left \rvert E\left[e^{i\hat{\theta}} \right] \right \lvert$. If one has a biased estimator, then $\mu = E\left[ \cos(\hat{\theta}- \theta) \right]$ \cite{Berry2009}. A circular estimator $\hat{\theta}$ is unbiased if $e^{i \hat{\theta}} \propto E\left[e^{i\hat{\theta}}\right]$. For narrowly peaked and symmetric distributions around $\theta$, $V^{H}_{M}(\hat{\theta}) \sim \text{MSE}_{M}(\hat{\theta})$. As a result,  Holevo's variance is lower bounded by
\begin{equation*}
  \label{QCRBCov}
  V^{H}_{M}(\hat{\theta}) \geq \frac{1}{F_Q(\theta)}\,. 
\end{equation*}

As the quantum Fisher information reaches its largest values for the families of pure states, let us restrict our attention only to pure quantum probes. Moreover, when the probe is a pure state, one can find the covariant POVM that minimizes the Holevo variance using the spectrum of $J$ for any shift parameter \cite{Holevo2011, Holevo2005}. Here, we adapt the proof of \cite{Holevo2011} to the problem of phase estimation. 
\begin{theorem}
  \label{Theo3}
  Let $ \rho _{\theta} = U_{\theta}\rho U_{\theta}^{\dagger}$, with $\rho = \ket{\psi}\bra{\psi}$ and 
 \begin{equation}
    \label{eq:MstarProy}
    M_{*}(d \hat{\theta}) = \frac{d \hat{\theta}}{2 \pi} \sum_{m,n \in Spec(J)}  e^{i \hat{\theta} \left(n-m\right) }  \cfrac{P_{m} \rho P_{n}}{\sqrt{\text{Tr}[P_m \rho]  \text{Tr}[P_n\rho]}}  ,
  \end{equation}
  where $P_m = \ket{m}\bra{m}$ is the associated projector to the eigenvalue $m$ of $J$. Then, for any $M \in \mathcal{M}\left(\Theta, U \right)$,
  \begin{equation}
    V_M^{H}(\hat{\theta}) \geq V_{M*}^{H}(\hat{\theta})\,.
  \end{equation}  
\end{theorem}

\begin{proof}
In Dirac notation, $\rho = \ket{\psi}\bra{\psi}$ so that Eq. (\ref{eq:MstarProy}) takes the following form
\begin{equation}
  \label{eq:Mstar}
    M_{*}(d \hat{\theta}) = \frac{d \hat{\theta}}{2 \pi} \sum_{m,n \in Spec(J)}  e^{i \hat{\theta} \left(n-m\right) }   p_{mn}^{*}     \ket{m}\bra{n},
\end{equation}
with $p_{mn}^{*} =\frac{\bracket{m|\psi}}{|\bracket{m |\psi}|} \frac{\bracket{\psi|n}}{|\bracket{\psi |n}|}$. On other hand, from Eq. (\ref{eq:mbasis}), we have   
  \begin{equation}    
      E_{M, \theta}[ e^{i k \hat{\theta}}] =  \sum_{\substack{m,n :\\ |m-n|=k}}  p_{m n} \bracket{\psi|m}  \bracket{n|\psi}  e^{-i(n-m)\theta}.
  \end{equation}
  Since $|p_{mn}| \leq 1$,
  \begin{equation}
    \begin{split}
      E_{M, \theta}[ e^{i k \hat{\theta}}] &\leq \sum_{\substack{m,n :\\ |m-n|=k}}  \bracket{\psi|m}  \bracket{n|\psi}  e^{-i(n-m)\theta}\leq \sum_{\substack{m,n :\\ |m-n|=k}}  |\bracket{\psi|m}|   |\bracket{n|\psi}|  e^{-i(n-m)\theta}.
      \end{split}
    \end{equation}
    where the equality is achieved if and only if $p_{m n}=p_{m n}^{*}$. Setting $k=1$ yields the assertion $\square$.
     \end{proof}
 Besides, if $\bracket{m | \psi}$ is a constant for all $m \in
 \text{Spec}(J)$, the  measurement $M_{*}$ has information about
 $\theta$ equal to the quantum Fisher  information \cite{Holevo2005}.
 For qubits, the previous condition is equivalent  to have an initial
 condition $\vec{a} \cdot \vec{n}= 0$.

 In this context, we now prove that the measurement $M_{*}$ maximizes
 the Fisher information over the set of covariant POVMs.

\begin{theorem} 
Let $\rho _{\theta} = e^{-i \theta \vec{n} \cdot \frac{\vec{\sigma}}{2}}\rho e^{-i \theta \vec{n} \cdot  \frac{\vec{\sigma}}{2}}$, with 
\begin{equation*}
\begin{split}
\rho = \ket{\psi}\bra{\psi} = \frac{1}{2}\left[ I + \vec{a} \cdot \hat{n} \right] \,,\quad\quad \norm{ \vec{a}}
  =1. \end{split}
\end{equation*}
 Then, for any $M \in \mathcal{M}\left(\Theta, U \right)$,
  \begin{equation}
    F(\theta; M) \leq F(\theta; M_{*})\,,
  \end{equation}  
where $M_*(d \hat{\theta})$ is defined according to Eq.  (\ref{eq:MstarProy}).
\end{theorem}

\begin{proof}
  Writing explicitly Eq. (\ref{eq:brspin}) we obtain   
  \begin{equation}
    \label{eq:probexplicit}
    p(d \hat{\theta}) = \frac{d \hat{\theta}}{2 \pi} \left[ 1 + 2 \text{Re} \left( e^{-(\hat{\theta}-\theta)} \cdot C \right) \right],
   \end{equation} 
where $C = p_{\footnotesize{\frac{1}{2} -\frac{1}{2}}} \bracket{-\tfrac{1}{2}|\psi}\bracket{\psi|\tfrac{1}{2}} $ is a complex number with module $|C|$ and phase $\varphi$. Thus,
  \begin{equation}
    \label{eq:probexplicit2}
    p(d \hat{\theta}) = \frac{d \hat{\theta}}{2 \pi} \left[ 1 + 2  |C| \cos\left( (\hat{\theta}- \theta) + \varphi  \right) \right].
   \end{equation}
   Therefore the Fisher information,
   \begin{equation}
     \label{FisherCova}
     \begin{split}
       F(M) &= \frac{1}{2 \pi} \int_{\Theta}  d\hat{\theta}  \frac{4|C|^2 \sin^2\left( (\hat{\theta}- \theta) + \varphi  \right) }{1+2|C|\cos\left(  (\hat{\theta} - \theta) + \varphi    \right)}
       = \frac{4|C|^2}{2\pi} \int_{\Theta} dx \frac{\sin^2\left( x \right)}{1+2|C|\cos\left( x \right)}.
       \end{split}
     \end{equation} 
The value of the expression given by Eq. (\ref{FisherCova}) increases monotonically as $|C|$ increases. As   $ \lvert \bracket{\pm \tfrac{1}{2} | \psi} \rvert =  \sqrt{ \tfrac{1}{2} \left( 1 \pm (\vec{a} \cdot \vec{n}) \right) }$, then   $|C| = |\leq \frac{1}{2} \sqrt{F_Q}$, and the maximum value for $F(M)$ is attained by $M_{*}$ $\square$.
\end{proof}
To find an explicit epxression of $p(d \hat{\theta}|\theta )$, we first notice that
  \begin{equation}
    \label{eq:projects}
    \begin{split}
\ket{\pm \tfrac{1}{2}} \bracket{\pm \tfrac{1}{2}|\psi}\bracket{\psi|\mp \tfrac{1}{2}} \bra{\mp \tfrac{1}{2}} = \frac{1}{4}\left[ 2(\vec{a} \cdot \vec{\sigma}) + i(\vec{a} \times \vec{n}) \cdot \vec{\sigma} - (\vec{n} \cdot\vec{\sigma})(\vec{n} \cdot \vec{a}) \right]\, ,      
      \end{split}
    \end{equation}
so that Eq. (\ref{eq:MstarProy}) can be rewritten as follows
  \begin{equation}
    \label{eq:mstarexp}
    M_{*}(d \hat{\theta}) = \frac{d \hat{\theta}}{2 \pi} \left[ I +  \frac{1}{\sqrt{F_Q}}    \left[   \left( \vec{a} \cdot \vec{\sigma} - (\vec{n} \cdot \vec{\sigma})(\vec{n} \cdot \vec{a}) \right) \cos(\hat{\theta}) -  \sin(\hat{\theta}) (\vec{a} \times \hat{n}) \cdot \vec{\sigma} \right]    \right]\,.
  \end{equation}
Then, the measurement $M_*$ yields the probability density function
  \begin{equation}
    \label{eq:pmstar}
    p(d \hat{\theta} \mid \theta ) = \frac{d \hat{\theta}}{2 \pi} \left[1 + \sqrt{F_Q}\cos(\hat{\theta}- \theta) \right]\,.
  \end{equation}
  As $M_{\star}$ is a POVM with an infinity number of outcomes, it is reasonable
  to get an identifiable statistical model independent of $\theta$. To
  illustrate this fact, we perform a numerical simulation, generating random
  numbers with distribution (\ref{eq:pmstar}). A particular likelihood function
  produced by (\ref{eq:pmstar}) is shown in Fig. \ref{Fig:iden}.
\begin{figure}[h!]
  \centering
    \includegraphics[scale = 0.6]{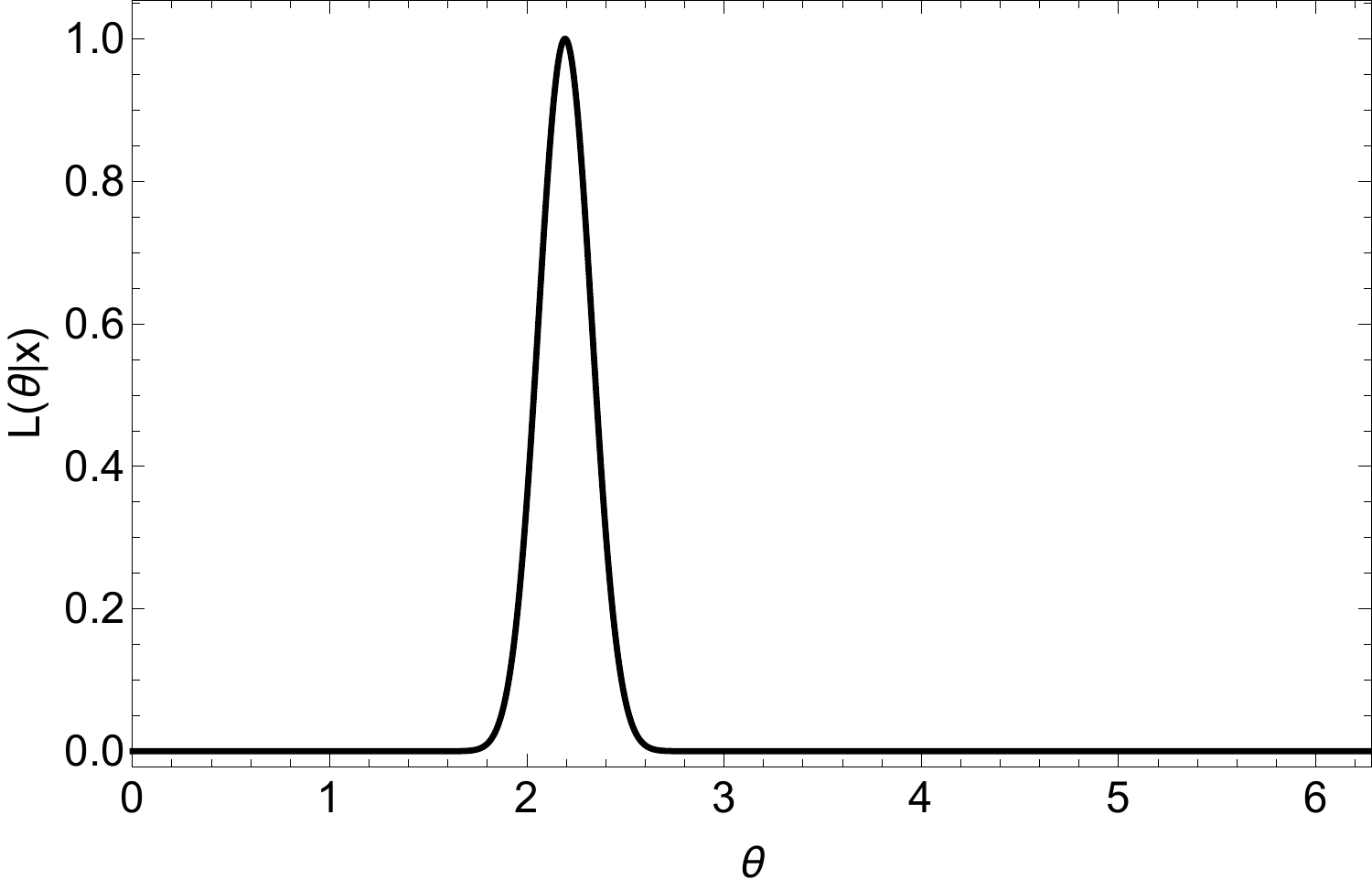}
  \caption{Plot of  the likelihood function corresponding  to the probability density Eq. \eqref{eq:pmstar} for the measure $M_{\star}$ and  $N=64$ probes. We have considered  the value of the actual parameter to be $\theta=2$ and used initial conditions $\vec{a}\cdot\vec{n}=0$. For clarity we have rescaled the figure so that the maximum equals to one.}
   \label{Fig:iden}
\end{figure} 

Moreover, according to Eq. (\ref{eq:pmstar}), the Fisher information given $M_{*}$ reads:
  \begin{equation}
    \label{eq:FIMstar}
       F(\theta; M_{*}) = \frac{1}{2 \pi} \int_{0}^{2 \pi} \frac{F_Q  \sin^2(\hat{\theta} - \theta)}{1+\sqrt{F_Q} \cos(\hat{\theta}- \theta)} d\hat{\theta}  = 1- \sqrt{1- F_Q} \,.
     \end{equation}
When $\vec{a} \cdot \vec{n}=0$, $F_Q=1$ and therefore $F(\theta; M_{*}) =  F^{\text{max}}_{Q}=1$. This is not surprising since, under this optimal initial condition, $\lvert \bracket{\tfrac{1}{2}|\psi} \rvert$ = $\lvert \bracket{-\tfrac{1}{2} | \psi} \rvert$. Consequently, assuming this optimal setup, if the maximum likelihood is used for $N$  independent copies of the system, one can attain the QCRB  asymptotically for large $N$. Nevertheless, under the optimal  initial condition, there are other covariant POVMs that can attain the QCRB. To see this, let us characterize all covariant measurements for quantum  phase estimation of qubits.

Every positive bounded operator $P_o$ can be expressed in the Bloch representation, that is, $P_o = d_0 + \vec{d} \cdot \vec{\sigma}$, where  $(d_0, \vec{d})$ are real. When evaluating  Eq. (\ref{eq:pc}), it turns out that $P_o$ forms a valid covariant POVM if $d_0=1$ and $\vec{d}$ is orthogonal to $\vec{n}$. Therefore any covariant POVM has the form
\begin{equation}
  \label{eq:covPOVM}
  \begin{split}
  M(d \hat{\theta}) &= \frac{d\hat{\theta}}{2 \pi} \left[ I + \vec{d} \cdot \vec{\sigma} \cos(\hat{\theta}) - \sin(\hat{\theta}) (\vec{d} \times \hat{n}) \cdot \vec{\sigma} \right]\,.
  \end{split}
\end{equation} 
When a qubit is in the state $\rho_{\theta}$, the outcomes for the POVM $M(d\hat{\theta})$ follows the  probability density
\begin{equation}
  \label{eq:probcov}
  \begin{split}
  p(d \hat{\theta} | \theta ) &= \frac{d \hat{\theta}}{2\pi} \left[1+ \left(\vec{a} \cdot \vec{d} \right)\cos(\hat{\theta}- \theta)  - \vec{a} \cdot (\vec{d} \times \vec{n}) \sin(\hat{\theta}- \theta) \right]\,.
  \end{split}
\end{equation}
From this result, one can calculate the Fisher information associated with the covariant POVM $M$. In particular, the case $\vec{d} = \pm(\vec{n} \times \vec{a})$ is of interest, since the classical Fisher information reads
\begin{equation}
  \label{eq:ficov}
   F\left(  \theta; \left. M  \right\rvert_{\vec{d} = \pm(\vec{n} \times
\vec{a})}  \right) = \frac{1}{2 \pi} \int_{0}^{2 \pi} \frac{F_Q^2  \cos^2(\hat{\theta} - \theta)}{1 \pm F_Q \sin(\hat{\theta}- \theta)} d\hat{\theta}\,.
\end{equation}
As far as we are aware of, equations (\ref{eq:covPOVM})-(\ref{eq:ficov}) 
are new results. 
Note that, in general, $F(\theta, M_{*}) \geq F\left( \theta; \left. M
  \right\rvert_{\vec{d} = \pm(\vec{n} \times \vec{a})} \right)$, but
for the particular case  $\vec{a} \cdot \vec{n} = 0$, the POVM $\left.
  M \right\rvert_{\vec{d} =\pm(\vec{n} \times \vec{a})} $ is able to
reach the maximum of the Fisher information, $F^{\text{max}}_{Q} = 1$.
Hence, the POVM given by Eq. (\ref{eq:covPOVM}) is equivalent to the
measurement $M_{*}$. \\

Let us briefly summarise the results we have found thus far. Firstly,
we have shown that estimation strategies based on locally optimal POVM
with two outcomes cannot achieve the QCRB since the corresponding
likelihood is non-identifiable. POVMs with more than two outcomes may
be identifiable but we could not find the POVMs that are locally
optimal (globally optimal POVMs do not exist
\cite{Barndorff-Nielsen2000}). Secondly, we have seen that for
continuous and periodic outcomes the covariant POVM solves the
identifiability problem and minimizes the MSE for one measurement,
however, it can only achieve the QCRB (in the limit of a large number
of measurements) with optimal initial conditions
$\vec{a}\cdot\vec{n}=0$. Hence, to the best of our knowledge, it does
not exist a set of
independent and identical measurements that achieve the QCRB.  \\

It turns out that, under the optimal set of initial conditions, the
covariant POVM $M_{*}$ has been previously investigated and it is
usually called \textit{canonical phase measurement} \cite{Holevo2005,
  Berry2009, Martin2019} which, according to Eq. (\ref{eq:Mstar}),
takes the following form:
\begin{equation}
  \label{eq:Mstar_can}
  M_{*}^{can}(d \hat{\theta}) := \frac{d \hat{\theta}}{2 \pi} \sum_{m,n \in Spec(J)}  e^{i \hat{\theta} \left(n-m\right) } \ket{m}\bra{n}\,.
  \end{equation}
  Although the optimal covariant POVM is hard to realize experimentally, it can
  be well approximated by POVMs with large number of elements or adaptive
  measurements \cite{Chapeau2016, Peng, Martin2019}. In particular, the
  canonical phase measurement can be implemented using adaptive measurements
  with quantum feedback \cite{Martin2019}, and the POVM $M_{*}$ can be well
  approximated using a POVM with $k \geq 8$ that discretizes Eq.
  (\ref{eq:pmstar}).

\subsection{Entangled measurement}
\label{lb:sb_entangled}
It is actually possible to achieve the QCRB for a probe not orthogonal to the rotation axis if we relax the condition of identical and independent set of measures. Let us indeed see it is possible to attain the QCRB for  entangled measurements for any initial condition. Following \cite{Holevo2005}, we estimate $\theta \in [0, 2 \pi)$ for the family of states $ \rho^{\otimes N}_{\theta}$ in the  Hilbert space $\mathcal{H}^{\otimes N}$  with $\rho_{\theta} =U_{\theta}\rho U_{\theta}^{\dagger}$ and $\rho$ being a probe pure qubit. Note that $\rho^{\otimes N}_{\theta}$ can be written as $\rho^{\otimes N}_{\theta} = e^{-i \theta J^{(N)}} \rho^{\otimes N} e^{i \theta J^{(N)}}$, where
\begin{equation}
  \label{eq:sn}
  J^{(N)} = \left(\hat{n} \cdot \frac{\vec{\sigma}}{2} \right) \otimes \cdots \otimes I + \cdots + I \otimes \cdots \otimes \left(\hat{n} \cdot \frac{\vec{\sigma}}{2} \right)\,,
\end{equation}
that is, the family of states $\rho^{\otimes N}_{\theta}$ in $\mathcal{H}^{\otimes N}$ is covariant with respect to the unitary representation $e^{-i\theta J^{(N)}} $. Therefore, as in the case of one probe, the POVM that minimizes the Holevo variance is expressed in terms of the spectrum of $J^{(N)}$.  Therefore, the measurement $M_{*}^{(N)}$ that minimizes the Holevo variance has the expression
\begin{equation}
  \label{eq:EntPOVM}
  \begin{split}
  &M_{*}^{(N)}(d \hat{\theta}) = \frac{d \hat{\theta}}{2 \pi}
  \left[ \sum_{\substack{\lambda, \lambda' \in\\ Spec\left(J^{(N)}\right) }} e^{i \hat{\theta}(\lambda'-\lambda)} \frac{P_{\lambda} \rho^{\otimes N} P_{\lambda'}}{
 \sqrt{  Tr\left[  P_{\lambda} \rho^{\otimes N} \right]    Tr\left[ P_{\lambda'} \rho^{\otimes N}\right]       }}  \right],
  \end{split}
\end{equation}  
where, $P_\lambda$ is the projection operator of $J^{(N)}$ associated to the eigenvalue $\lambda$.\\
The difference here, in contrast to the single measurement case, is that here the spectrum of $J^{(N)}$ is degenerate. Thus
\begin{equation*}
  P_{*}= \int_{\Theta} M(d \hat{\theta}) \neq I,
\end{equation*}  
and therefore $M_{*}^{(n)}$ is not a resolution of identity. However, this is not a real problem because it can be extended to a POVM by adding $I-P_{*}$ in the orthogonal complement to the subspace $\mathcal{H}_{*}:=P_{*}(\mathcal{H}^{\otimes N})$. So, the optimal POVM is
\begin{equation*}
M_{*}^{(N)}:= P_{*} \oplus I-P_{*}.
\end{equation*}  
In particular, if $N=1$, $M_{*}^{(N)} = M_{*}$.\\
As stated in \cite{Holevo2005}, $M_{*}^{(N)}$ can achieve the QCRB for any probe state in the asymptotic limit. Besides, the scaling of the error is proportional to the square in the number of probes (Heisenberg scaling).  Unfortunately, $M_{*}^{(N)}$ is an entangled measurement. Hence, if one wants to construct an entangled POVM, it is necessary to have $N$ copies of the probe and perform quantum operations over all of them,  making the implementation of this type of POVMs an experimental challenge.

\subsection{Adaptive state quantum estimation (AQSE)}
\label{lb:sub_aqse}
As we mentioned in the introduction, unlike its classical counterpart, quantum
parameter estimation yields estimators which depend on the parameter one wants
to estimate. A way to tackle this problem is using an adaptive quantum
estimations scheme \cite{Fujiwara2011, Nagaoka,Yamagata2013AsymptoticQS}. It
works as follows: suppose one has a set of optimal POVMs $\left\{ M_g
\right\}_{g \in \Theta}$. One begins with an arbitrary initial guess $g_0$. Then
one applies the optimal measurement at $g_0$, $M_{g_0}$. Assuming the data $x_1$
is observed, one applies the MLE to the likelihood function $L_1(\theta \mid
x_1; g_{0}) = p(x_1 \mid \theta; g_{0})$ to obtain an estimate
$\hat{\theta}_1(x_1) =g_1$. This process is then repeated iterative. For $n \geq
2$, one applies the POVM $M_{g_{n-1}}$, with $g_{n-1}=
\hat{\theta}_{n-1}(x_1,...,x_{n-1}))$, obtaining the outcome $x_n$. Here
$\hat{\theta}_{n-1}$ is the estimation from the previous step using the outcomes
$x_1,...,x_{n-1}$. The likelihood function for the $n$th step is
\begin{equation}
  \label{eq:like_AQSE}
  L_{n}(\theta \mid x_1,...,x_{n-1}; g_{n-1}) = \prod_{i=1}^{n} p(x_i \mid \theta; g_{i-1})\,,
\end{equation}
where $x_i$ is the data observed at step $i$. Applying the MLE one
obtains the $n$th guess $g_n = \hat{\theta}_{n}\left( x_1,...,x_n\right)$. Even though this method relaxes the condition of identical measurements, this  does not guarantee that the resulting classical statistical model is regular thus impacting its performance \cite{Dee}.\\
To illustrate this let us revisit the problem of non-identifiability for the POVM $M_g$ previously discussed. To fix ideas we take $\theta=2$ and consider $N=64$ measurements. In  panel A of Fig~\ref{Fig:NI} we show the result of the likelihood function for $g=1.5$ using Eq. (\ref{eq:p_mi}). The likelihood is indeed non-identifiable and presents two global maxima in the interval  $[0, 2\pi)$. Panel B in the same figure shows the results of using AQSE starting with the initial value $g_0=1.5$. In this case we can see that there is only one global maxima but at the incorrect value of $\theta$. This is the result of having an original non-identifiable model, whose wrong maximum has been enhanced due the particular  random trajectory in the parameter space generated by AQSE. Thus, even though AQSE is able to lift a possible degeneracy of the global maxima there is no a priori way to control on how to enhance the correct maximum.\\
\begin{figure}[h!]
\centering
\includegraphics[scale = 0.47]{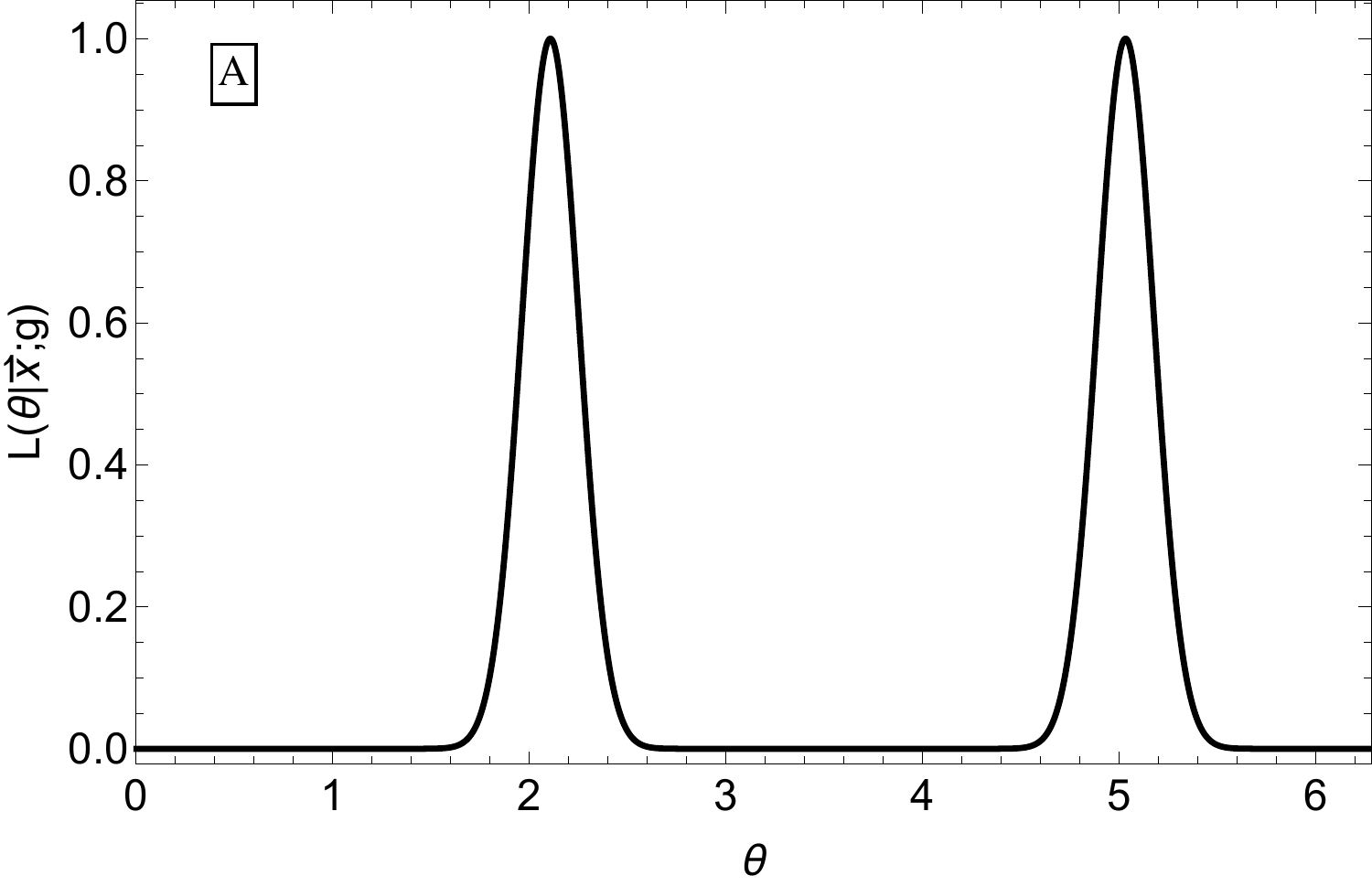}\quad\includegraphics[scale = 0.47]{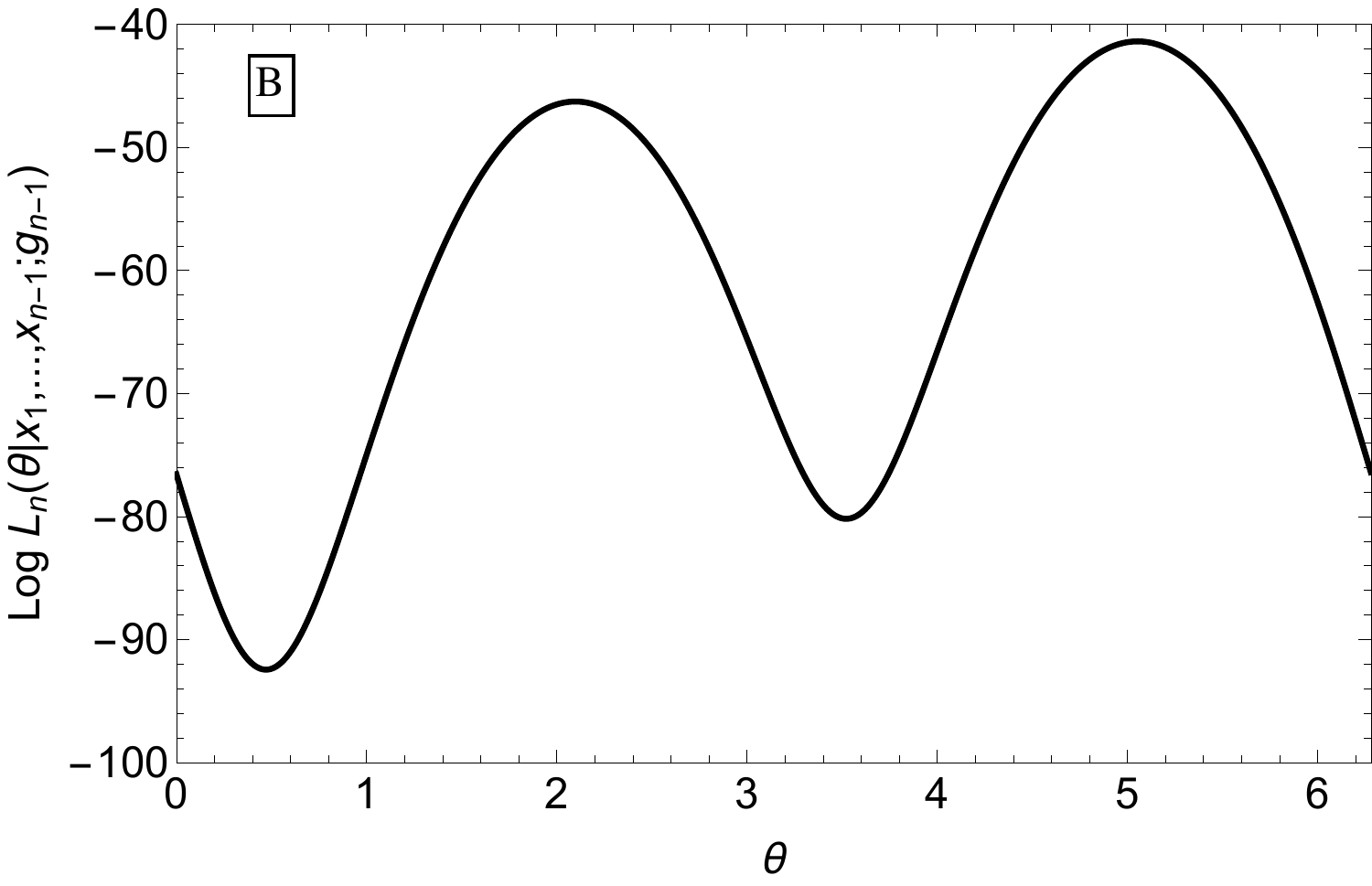}
\caption{Plot of the likelihood functions for the  two-outcome POVM $M_g$ given by Eq. \eqref{eq:projqs} and considering  $\theta =2$ and $N=64$ probes. Panel (A) shows the likelihood function given by formula \eqref{eq:likeP}. Here we have considered $g=1.5$ and the number of ones was 32  out of the possible $N=64$. For clarity the figure has been rescaled so that the maxima equal to unity. Panel (B) shows the log-likelihood function for the AQSE scheme given by Eq. \eqref{eq:like_AQSE} using $64$ adaptive steps. The existence of trajectories in the parameter space enhancing the incorrect maximum explains why the AQSE does not reach the QCRB. In both cases, the initial condition has been taken to be $\vec{a} \cdot \vec{n} = 0.5$.}
  \label{Fig:NI}
\end{figure}
With a modest amount of foresight, it seems  fairly natural to expect
that if we knew in which interval the parameter lies and  restricted the likelihood over that interval -so as to ensure that the restricted likelihood is identifiable- then the adaptive scheme will converge to the QCRB fast. To analyse this we have compared the behaviour of AQSE in two different scenarios: unrestricted and restricted likelihoods. In both cases we have taken $\theta=\pi$. For the unrestricted scheme the parameter space is  $\Theta = [0, 2\pi)$, while for the restricted one we take $\theta \in \Theta = [\frac{\pi}{2}, \frac{3 \pi}{2})$, so that the likelihood is identifiable. In Fig.~\ref{Fig:S1} we show the result of these two cases for two extremal initial conditions $\vec{a} \cdot \vec{n} = 0.5$ (panel A) and  $\vec{a} \cdot \vec{n} = 0$ (panel B), corresponding to the lowest and largest value of  $F_Q$, respectively. In both cases the numerical method was done using a bootstrap simulations with $10000$ repetitions. \\
From these results we can observe that in the unrestricted case the lack of identifiability  directly affects the scaling of the Holevo variance as a function of the number of adaptive steps. Furthermore, numerical results show that AQSE does not saturate the QCRB. On the other hand, in the restricted case, we find that the MLE saturates the QCRB around $128$ measurements for both initial conditions. Thus we conclude that a possible deficient performance of  AQSE is due to the non-identifiability for $\theta$ in the likelihood functions.\\
\begin{figure}[h!]
    \centering
   \includegraphics[scale=0.47]{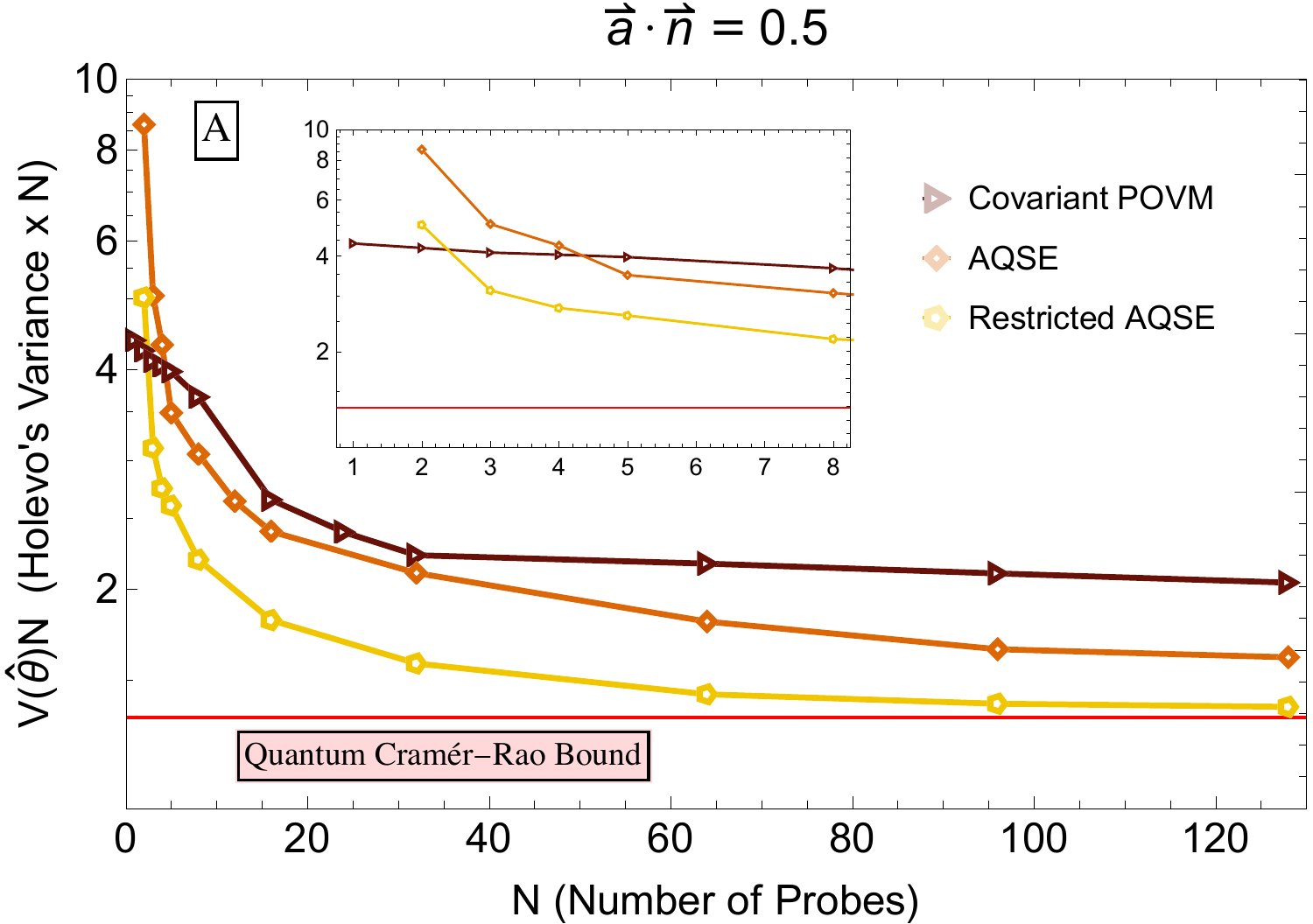}\quad  \includegraphics[scale=0.47]{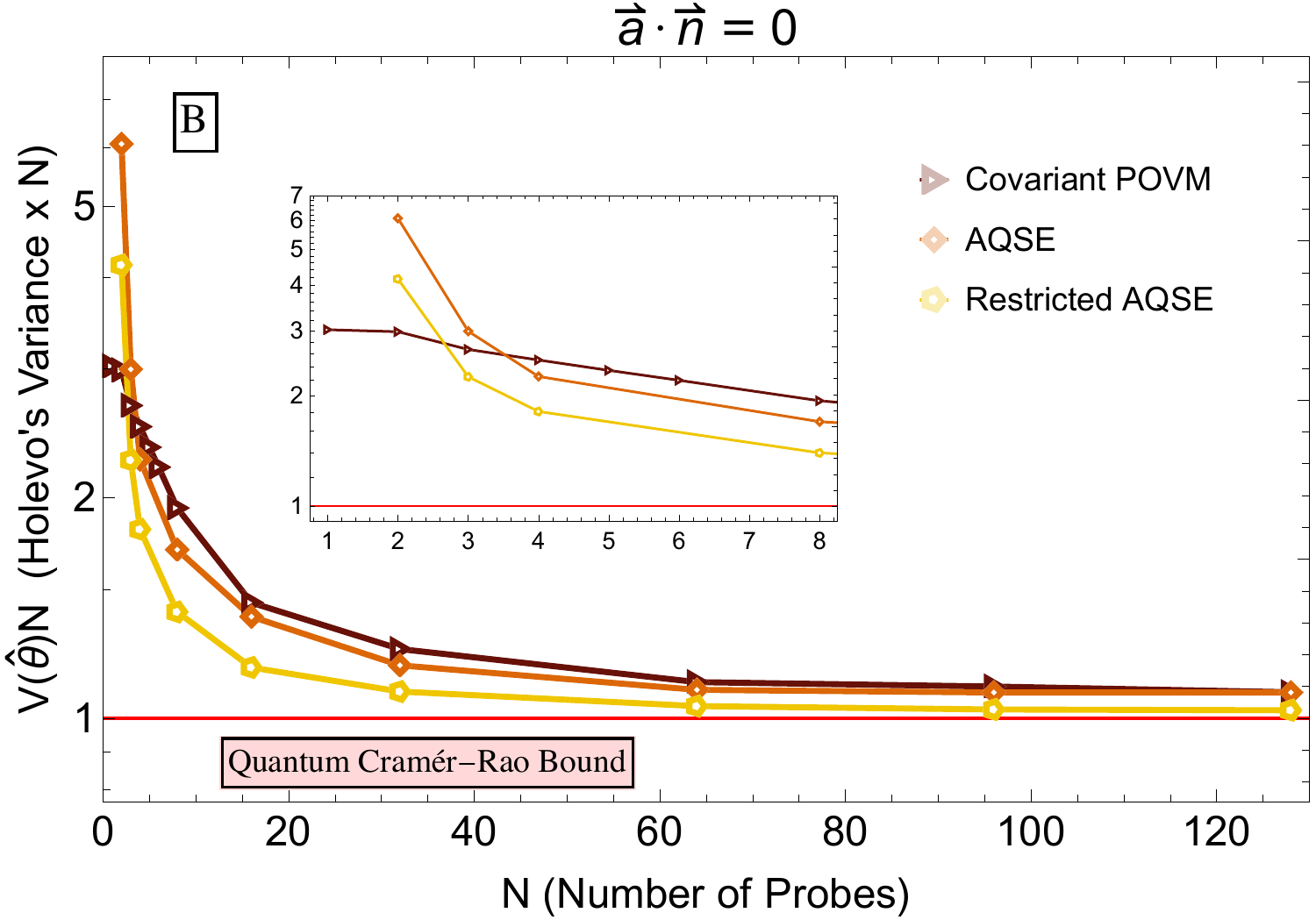}
  \caption{Result of Holevo's variance as a function of the number of probes for  different estimation strategies and for two different initial conditions:  $\vec{a}\cdot\vec{n}=0.5$, corresponding to the smallest quantum Fisher information (shown in panel A) and  $\vec{a}\cdot\vec{n}=0$ (shown in panel B). In all cases we used $10000$ sequences for each number of probes. As we can see the  restricted AQSE performs much better than the AQSE. We have also included the strategy based on the optimal covariant POVM $M_{*}$ as a benchmark. Notice that for both initial  conditions the AQSE does not reach the QCRB, while for optimal initial condition  the sequence of covariant POVMs is able to  attain the QCRB (see panel B).}
  \label{Fig:S1}
\end{figure}
Finally, in Fig.~\ref{Fig:S1}, we have also compared the performance of AQSE with the optimal covariant inference. In this case,  the sequence of $M_{*}$ measurements has a better scaling for a small number of measurements (less than $4$). Then, we can also conclude that the optimal covariant measurement is the best independent strategy in a regime of small $N$. This is expected because the covariant measurement is built to minimize the MSE for one measurement.\\
In the next section, we propose an estimation scheme that uses the covariant estimation and the AQSE method to avoid the problem of non-identifiable likelihood functions and thus is able to saturate the QCRB for any initial condition.

\section{Adaptive estimation scheme with confidence intervals}
\label{lb:arbitrary_initial_condition}
From the previous discussion it seems clear that, conditioned of knowing in which interval the actual phase lies in and restricting the likelihood to that interval so it becomes identifiable, the AQSE method converges efficiently to the QCRB. With this in mind, we propose the following two-step scheme which allows us to initially guess the interval, thus making  the iterative scheme a regular problem. \\
Our scheme consists of two main steps. In the first step one uses a sequence of  independent optimal covariant measurements to produce a confidence interval (CI) using the maximum likelihood estimate $\hat{\theta}_{\text{MLE}}$. The CI  gives a range of plausible values where the unknown parameter most likely lies in. In the second step, one applies the AQSE method restricted to the confidence interval. Furthermore, one is able to update the center of the CI by the MLE of each previous step. The idea of using CIs to improve the error in estimations is exemplified, in a different context, in \cite{Boxio}.

Let us assume that we have $N$ copies of the state system $\rho$, so that each measurement is performed consecutively over each copy. Let us denote the result of each measurement as $x_i$. In the first part of our scheme we apply $N_1$ covariant measurements $\mathcal{M}_1$ given by Eq.  (\ref{eq:Mstar}), that is
\begin{equation}
  \label{eq:propPOVM1}
  \mathcal{M}_1 := M_{*}(d x_k) \, ,
\end{equation}
and then we construct the MLE
\begin{equation}
  \label{eq:mle_cov}
  \hat{\theta}_{\text{MLE}} = \text{arg max}_{\theta \in \Theta}  \prod_{i=1}^{N_1} p \left( d x_i \mid \theta \right) = \text{arg max}_{\theta \in \Theta}  \prod_{i=1}^{N_1}   \frac{d x_i}{2 \pi} \left( 1 + \sqrt{F_Q} \cos( x_i - \theta) \right) 
\end{equation}
to obtain the estimation $\hat{\theta}_{\text{MLE}}\left( x_1,...,x_{N_1} \right)$. The CI is then given by  \cite{degroot}:
\begin{equation}
  \label{eq:CI}
  {\rm CI}(\hat{\theta}_{\text{MLE}}(\vec{x})) = \left(\hat{\theta}_{\text{MLE}}(\vec{x}) - c \cdot
    F^{-\frac{1}{2}} , \, \hat{\theta}_{\text{MLE}}(\vec{x}) +c \cdot
    F^{-\frac{1}{2}} \right)\, ,
\end{equation}
where $c$ is the appropriate critical value in the standard normal distribution (e.g. $1.96$ for $95\%$ of confidence or $2.58$ for $99\%$ of confidence) and $F$ is the Fisher information of $M_{*}$ given by Eq.~(\ref{eq:FIMstar}). To determine the minimum sample size $N_1$ in order to get a confidence interval of size $2E$, we use the CRB to write the lower bound
\begin{equation}
  \label{eq:SS}
 N_1\geq \frac{c^2}{F E^2} \,.
\end{equation}
In the second step of our scheme we then use the remainder of the copies, that is  $N_2=N-N_1$, to perform  measurements  using the POVM
\begin{equation}
  \label{eq:propPOV2}
  \mathcal{M}_2 := M_{\hat{\theta}_{\text{MLE}}\left( x_1,...,x_{k} \right) } \,.
\end{equation}
All in all, our two-step method is given by:
\begin{equation}
  \label{eq:mle_prop}
  \hat{\theta}_{\text{MLE}} = \text{arg max}_{\theta \in \Theta} \prod_{i =1}^{N_1} p(d x_i \mid \theta ) \prod_{k=N_1}^{N_2} p( x_{k+1} \mid \theta; \hat{\theta}_{\text{MLE}}(x_1,...,x_k) )\, ,
\end{equation}
with  $\Theta = \mathrm{CI}\left( \hat{\theta}_{\text{MLE}}\left( x_1,...,x_{k}\right) \right)$. Eq. \eqref{eq:mle_prop} is the main result of this work.\\
Before discussing some numerical results of our newly introduced method it is important to discuss the possible sources of errors associated to it. Suppose that we set a confidence level $0\leq C_l\leq 1$ for a given marginal error $E$. Then we can write that
\begin{equation}
  \label{eq:P_error}
  V^H \left( \hat{\theta}_{\text{MLE}} \right) = C_l  \Delta_1 +\left( 1 - C_l \right) \Delta_2\,,
\end{equation}
where the two types of errors $\Delta_1$ and $\Delta_2$ correspond to the error associated to the AQSE method when the CI correctly includes the value of the actual parameters, and when it does not, respectively. Let us call the latter intervals  bad CIs. Further, let us take $E< \frac{\pi}{2}$, so that  likelihood functions used in the AQSE method are identifiable. Thus, the error $\Delta_1$ when the CI includes the value of the parameter is lowered bounded by
\begin{equation}\label{eq:error_delta1}
\Delta_1 \geq \frac{1}{F_Q N_2 + N_1 F(\theta; M_*)}\, .
\end{equation}
To characterize $\Delta_2$, first we note that there are two types of bad CIs. To see this, let $\theta_A$ be the parameter's value and $\theta_B = \theta_A + \pi$. Then, the likelihood function produced by AQSE has two local maxima, one around $\theta_A$, and the other around $\theta_B$. The first type of bad CIs corresponds to those that include the second maximum around $\theta_B$. If one applies AQSE restricted to these CIs then $\theta_{\text{MLE}}$ tends to $\theta_B$. The second type of bad CIs appear when the interval does not include $\theta_B$. In this case, the maximum likelihood estimation produced by AQSE can tend to the point in the interval closest to either $\theta_A$ or to $\theta_B$. When the estimation is closest to $\theta_A$, by updating the  interval's center with the posterior estimates,  one can, in principle, obtain a confidence interval that captures the parameter's value. This explain why we obtain better results updating the center of the confidence interval.

To obtain a bad interval of type one, it is necessary that most of the data obtained from the covariant inference be numbers close to $\theta_B$ (the likelihood function has a global maximum around $\theta_B$).  From Eq.  (\ref{eq:pmstar}), for a covariant sample of size $N$, it follows that the probability of obtaining a likelihood function with the global maximum around
$\theta_B$ is
\begin{equation}
  \label{eq:p_2bci}
  \left[ \int_{\theta_B-\epsilon}^{\theta_B+\epsilon} p\left( d \hat{\theta} \mid \theta \right) \right] ^{N}  = \left[ \int_{\theta_B-\epsilon}^{\theta_B+\epsilon} \frac{d \hat{\theta}}{2 \pi} \left( 1 + \sqrt{F_Q} \cos(\hat{\theta}- \theta) \right) \right] ^{N}\,.
\end{equation}
Thus, for sufficiently large $N$, the above probability can be neglected. In this way, for a confidence level close to $1$, and a sufficient small marginal error $E$, one expects that the number of bad confidence intervals disappear as the number of adaptive steps increases. To illustrate this fact, in Table \ref{tabla2}, we show the number of bad intervals as a function of the adaptive steps in the AQSE part produced by a bootstrap simulation of $10000$ repetitions. Here we have chosen a marginal error of $\frac{\pi}{4}$, with a confidence level of $0.99$. In this case, we need to perform $11$ and $22$ covariant measurements for the initial conditions $\vec{a} \cdot \vec{n}=0$ and $\vec{a} \cdot \vec{n} = 0.5$, respectively, to get the desired CI. For these values, the number of type one bad confidence intervals is negligible. Setting $\epsilon = E$ in (\ref{eq:p_2bci}), the probabilities of obtaining them are $2.3 \times 10 ^{-10}$ and $2.3 \times 10^{-18}$ for these two  initial conditions, respectively.\\
\begin{table}[h]
\begin{center}
\begin{tabular}{|c|c|l|l|c|l|l|}
  \hline
  \multicolumn{1}{|l|}{\textbf{AQSE Steps}} & \multicolumn{3}{c|}{\textbf{\# bad CIs ($\vec{a}\cdot \vec{n}=0$)}} & \multicolumn{3}{c|}{\textbf{\# bad CIs ($\vec{a}\cdot \vec{n}= 0.5$)}} \\ \hline
  0                                            & \multicolumn{3}{c|}{111}                            & \multicolumn{3}{c|}{125}                                \\ \hline
  4                                            & \multicolumn{3}{c|}{50}                             & \multicolumn{3}{c|}{56}                              \\ \hline
  8                                            & \multicolumn{3}{c|}{19}                             & \multicolumn{3}{c|}{20}                               \\ \hline
  16                                            & \multicolumn{3}{c|}{8}                              & \multicolumn{3}{c|}{1}                                \\ \hline
  32                                            & \multicolumn{3}{c|}{6}                              & \multicolumn{3}{c|}{1}                                \\ \hline
  48                                            & \multicolumn{3}{c|}{2}                              & \multicolumn{3}{c|}{0}                                \\ \hline
\end{tabular}
\caption{Number of bad confidence intervals as a function of adaptive steps for
  the AQSE part where the real value of $\theta$ was $\pi$.\label{tab:Fujiwara}}
\label{tabla2}
\end{center}
\end{table} 
Note that,  the distribution for the estimates can be approximated as a mixture of $3$ normal distributions, $\mathcal{N}(\theta, \frac{1}{F_Q N_2 + N_1 F(\theta;  M_*)})$, $\mathcal{N}(\theta + E, \frac{1}{F(\theta, M_*,M)})$, and $\mathcal{N}(\theta - E, \frac{1}{F(\theta, M_*,M)})$, where $M$ is the sequence of AQSE and $E$ is the marginal error. This implies that the lower bound for the Holevo variance, when the center of the CI  is not updated, is given by
\begin{equation}
  \label{eq:BE_prop}
  \begin{split}
    V^H \left( \hat{\theta}_{\text{MLE}} \right) &\geq \frac{1}{F_Q N_2 + N_1 F(\theta;  M_*)} + (1-C_l) E^2\,.
  \end{split}
\end{equation}
The second term on the right-hand side of Eq. (\ref{eq:BE_prop}) dominates when $N_2 \to \infty$. This has a very simple interpretation: if the parameter is outside the confidence interval, making more steps in the AQSE method does not diminish the error. Unfortunately, we are not able to provide a lower bound to the Holevo variance in the case for which the center of the CI is also updated. Nevertheless, as we will see, by setting a confidence level close to $1$, a sufficiently small marginal error $E$, and by updating the center of the CIs, the Holevo variance  approximates the QCRB in a small number of steps. \\
To assess the performance of our method we have performed a Monte Carlo simulations by setting  $E = \frac{\pi}{4}$ and $C_l = .99$. Moreover, we have compared our numerical results with the ideal case ($C_l= 1$) and an estimation strategy for which the center of the CIs is not updated. The results are summarised  in Figures \ref{Fig:prop} and \ref{Fig:prop2} for initial conditions $\vec{a} \cdot \vec{n} = 0$ and $\vec{a} \cdot \vec{n} = 0.5$.\\
In Fig. \ref{Fig:prop} we have considered the initial condition  $\vec{a} \cdot \vec{n} = 0$, so that  $F(\theta; M_*) = 1$, and fixed  $N_1 = 11$. The vertical dashed line at  $N_1 = 11$ separates the point between covariant and AQSE inferences. Notice that in the covariant part, for $N\leq 11$,  both strategies corresponding to either updating or not the CIs centers yields the same result. However, for $N>11$, if the center of the CI is not updated, so that according to  Eq. (\ref{eq:BE_prop}) it second term dominates,  the estimation error grows quickly as the number of measurements increases. If, however, the center of the CIs is updated then the estimation error tends to the QCRB. Similarly  in Fig. \ref{Fig:prop2}, for which we have chosen $N_1 = 22$, $\vec{a} \cdot \vec{n} = 0.5$ and $F(\theta; M_*) = 0.75$, we can observe the same qualitative behaviour,  showing that our estimation strategy tends to the QCRB as the number of steps increases for any set of initial conditions.\\
\begin{figure}[h!]
  \centering
  \includegraphics[scale=0.6]{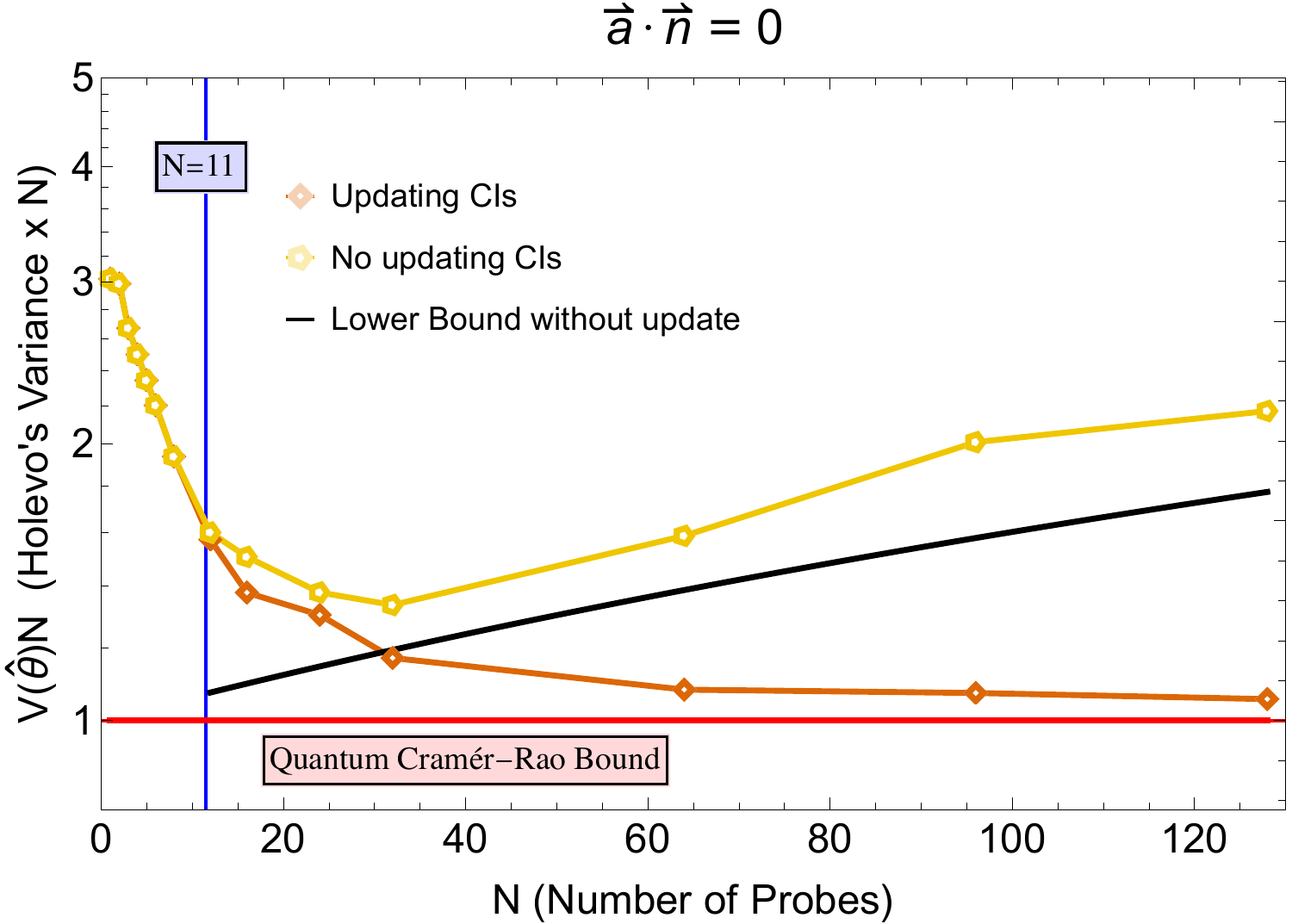}
  \caption{Plot showing Holevo's variance as a number of probes for our scheme. The optimal strategy (orange  rhomboid markers), in which we update the centers of the CIs, tends to the QCRB as the number of probes increases. The latter bound corresponds to take a confidence level of $100\%$. We also show the alternative strategy for which the centers of the CIs are not updated (yellow pentagon markers). This result is lowered bounded by the solid black line, whose formula is given by Eq. \eqref{eq:BE_prop}. We have chosen the initial condition  $\vec{a} \cdot \vec{n} = 0$. Here the vertical blue line at  $N=11$ marks the point separating the covariant inference performance from the AQSE inference to our two-step scheme.}
  \label{Fig:prop}
\end{figure}
\begin{figure}[h!]
  \centering
  \includegraphics[scale=0.6]{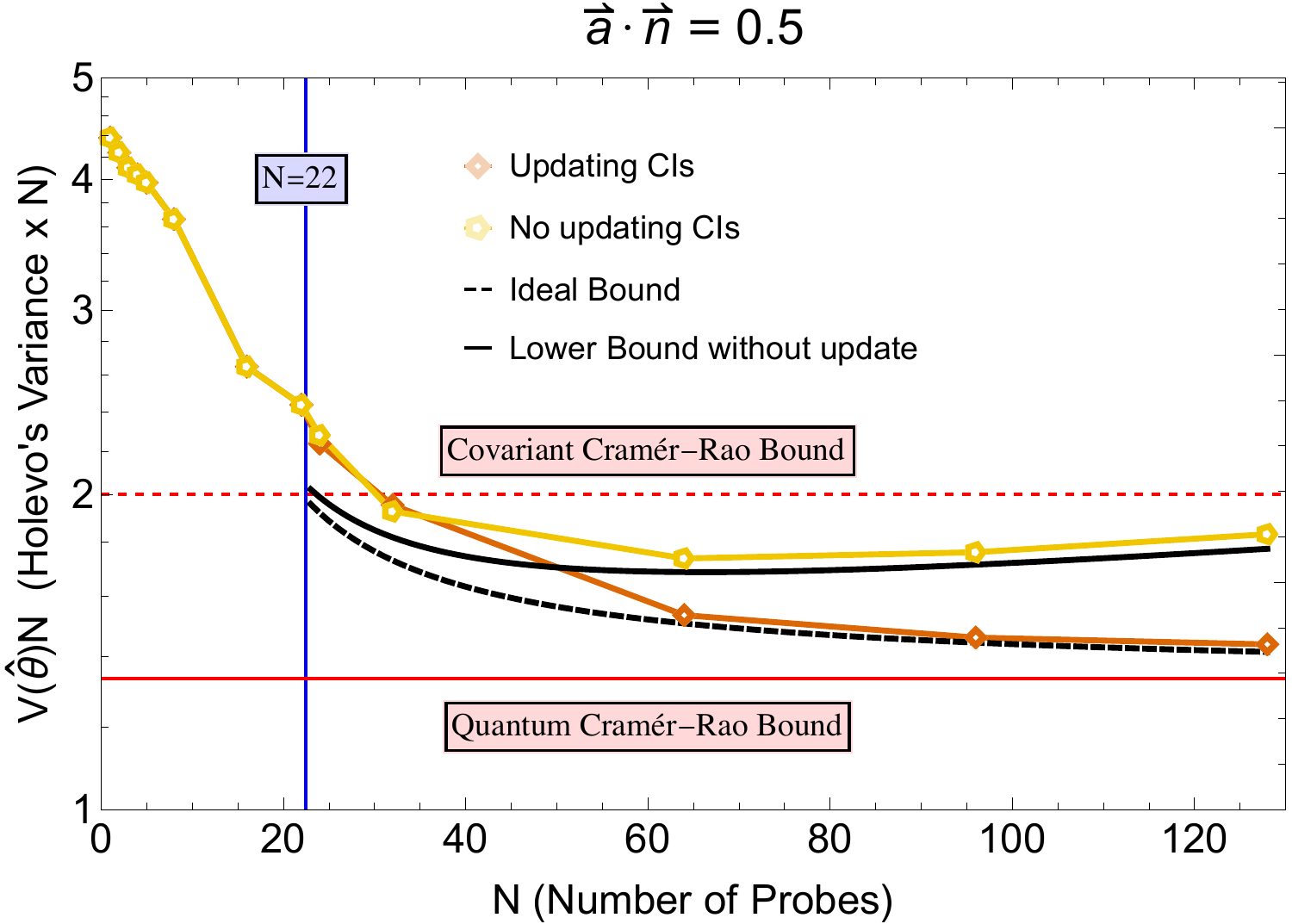}
  \caption{Plot showing Holevo's variance as a function of the number of probes $N$ for our  two-step adaptive scheme. The optimal strategy (orange  rhomboid markers), in which we update the centers of the CIs, tends to the QCRB as the number of probes increases. This result is lowered bounded by the dashed black line, which is the QCRB corrected using Eq. \eqref{eq:BE_prop} for a confidence level of $100\%$. We also show the alternative strategy for which the centers of the CIs are not updated (yellow pentagon markers). This result is lowered bounded by the solid black line, whose formula is given by Eq. \eqref{eq:BE_prop}. In this case, we have chosen the initial condition  $\vec{a} \cdot \vec{n} = 0.5$. Here the vertical blue line at  $N=22$ marks the point separating the covariant inference performance from the AQSE inference to our two-step scheme. Finally, the horizontal dashed red line corresponds to the Cram\'er-Rao bound when only covariant measurements are used instead.}
  \label{Fig:prop2}
\end{figure}
We conclude by showing in Fig.  \ref{Fig:P1} the performance of all the methods we have discussed throughout this work considering the initial conditions $\vec{a} \cdot \vec{n} =0.5$ and $\vec{a} \cdot \vec{n} = 0$, which correspond to the smallest and largest quantum Fisher information, respectively.  Obviously the optimal performance for the different estimation strategies is independent of $\theta$, since the Holevo variance is invariant under translation modulo $2\pi$ over the parametric space. For all the strategies, the behavior is qualitatively the same for any initial condition. As expected, the entangled strategy is the best one and reaches the QCRB. The second best strategy is the restricted ASQE, but this  is unrealistic as it assumes that one knows beforehand an interval where the MLE is regular and includes the parameter. Our method, which we emphasize is the more realistic both mathematically and experimentally, is the third best strategy, as it reduces the error of the estimates improving the performance of AQSE and approximating the estimate to the QCRB. 

\begin{figure}[h!]
  \centering
  \includegraphics[scale=0.47]{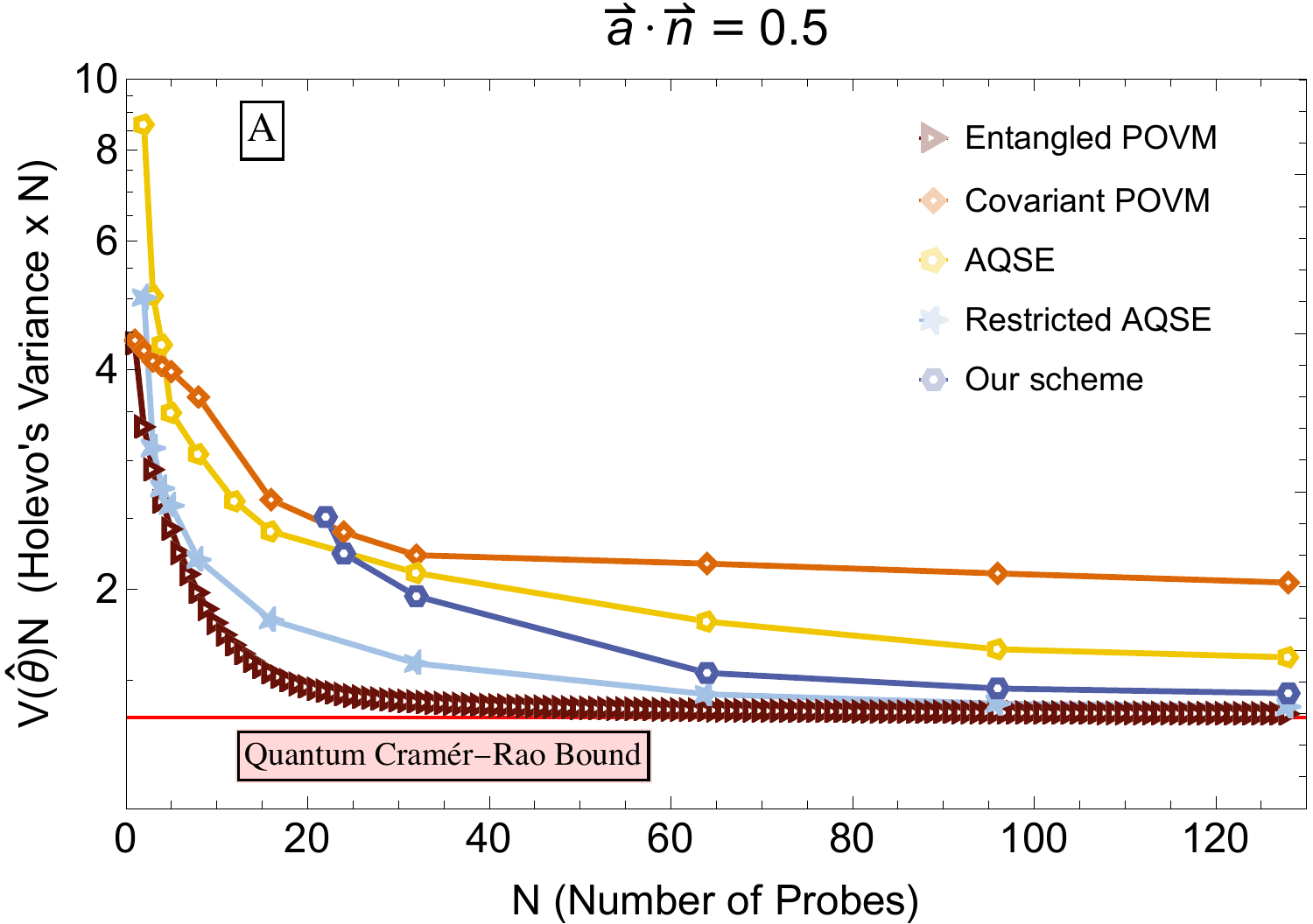}\quad \includegraphics[scale=0.47]{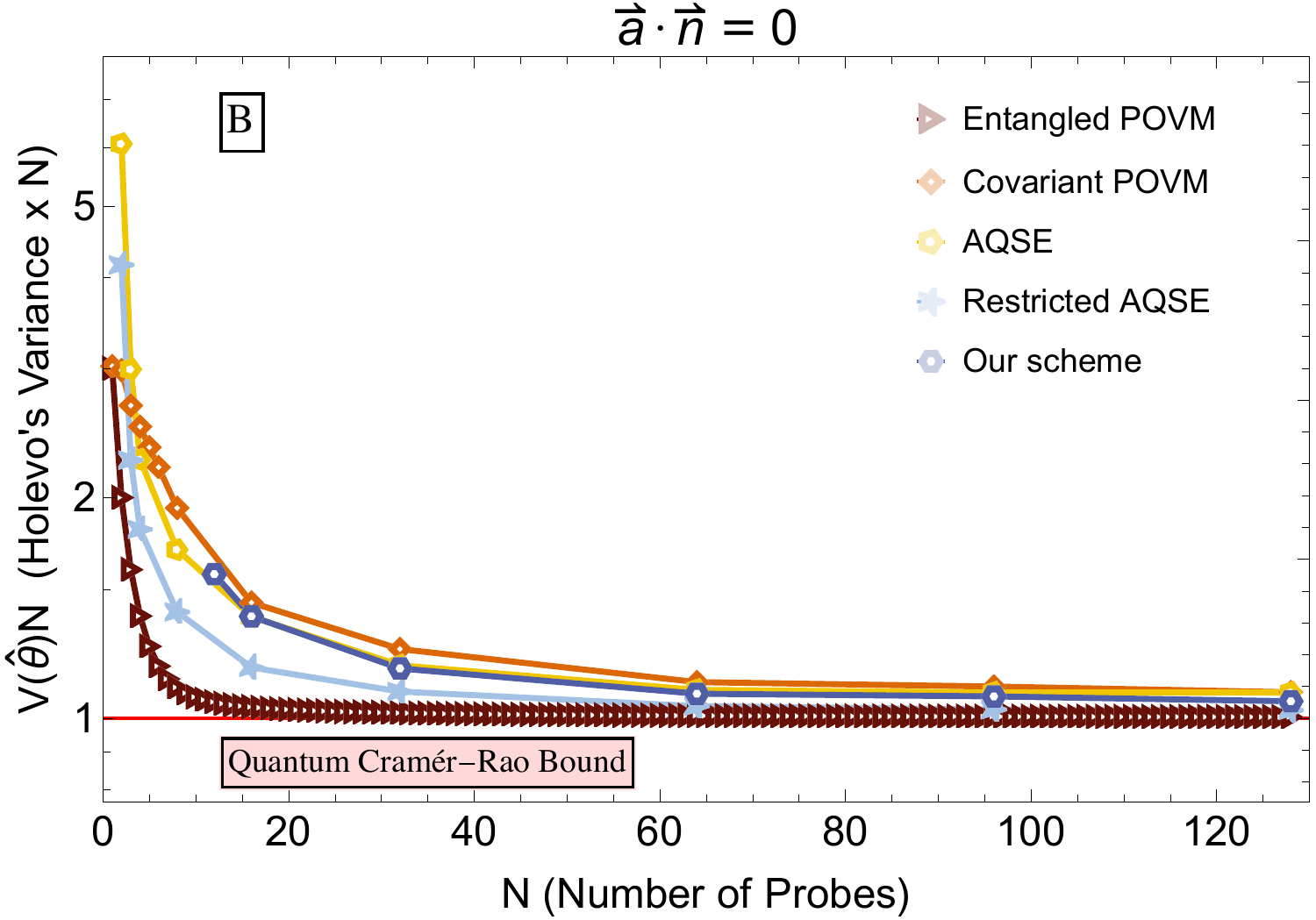}
  \caption{Holevo variance vs the number of probes for different
    strategies of estimation. In order to calculate the Holevo
    variance for the schemes based on the maximum likelihood
    estimator. A sequence of $10000$ measurements were simulated in
    each point. The curve for the entangled measurement was
    analytically calculated. The curve for the proposed scheme shows
    the performance with a $99\%$-confidence intervals. The restricted
    AQSE method assume that the parametric space is an interval of
    length $\pi$ that includes the real value of $\theta$. The best
    strategy is the entangled measurement, followed by the restricted
    AQSE. These two strategies are unrealistic to implement. The third
    best strategy is our proposal. The y-axis is in log scale. }
  \label{Fig:P1}
\end{figure}

\section{Summary and Conclusions}
\label{summary_conclusions}
In this work we have thoroughly examined several strategies for
quantum phase estimation. Their common denominator relies on
maximizing likelihood density functions which are generally
non-identifiable and/or non-optimal. This implies that the QCRB cannot be attained. We have developed a two-step strategy that circumvents this problem. Our method relies first on covariant measurements to identify a confidence interval within which the actual parameter is most likely to be, and  then to apply an adaptive technique restricted to that interval. When compared with the current existing methods, ours is  mathematically more robust, as it reaches the QCRB for any initial condition,  and experimentally is more realistic, since neither an entangled measurement  nor a priori information of the parameter's value are needed. Finally, we believe that our scheme can be generalized to any system so long as one can use a set of measurements to construct confidence intervals.

However, based on the results presented here and the current body of
work, a number of questions remain open. First of all, it is not clear
to us whether there exists a better strategy for which the subset of
covariant measurements is intertwined with the subset of adaptive ones
so as to identify the CI much faster. Secondly, the generalization to
this work to multiparameter estimation does not seem straightforward,
since although it is indeed possible to construct the covariant POVMs
\cite{Holevo2011} for this case, little is known about
non-identifiability for multiparameter likelihood functions. These,
and other issues, will be addressed in forthcoming works.

\section{Acknowledgments}
\label{Acknowledgments}
We thank Laboratorio Universitario de Cómputo de Alto Rendimiento
(LUCAR) of IIMAS-UNAM for their service on information processing.
This research was supported by the Grant No. UNAM-DGAPA-PAPIIT
IG100518 and IG101421, and Doctoral scholarship CONACYT 334231. I.P.C.
acknowledges funding support from the London Mathematical Laboratory
where he is an External Fellow.

\appendix

\section{Code implementation}
\label{reproducibility}
Our numerical results have been implemented using R language. An R library with the various methods discussed here can be publicly found in the repository \cite{Git_Repo}. To reproduce the numerical points for covariant strategy, AQSE and restricted-AQSE presented in Fig. \ref{Fig:P1}, use the function:
\begin{verbatim} 
Hvar_Scheme(theta_real, par_space, n_boost, num_prob, n,a, strategy),
\end{verbatim}
set the value of $\theta$, \verb+theta_real+ from $0$ to $2 \pi$, set \verb+par_space+ as the vector \verb+c(0, 2*pi)+, set the size of bootstrap \verb+n_boost=1000+, vary the desired number of probes \verb+num_prob+, set the initial unitary vectors for the probe and the axis of rotation \verb+n+ \verb+a+, and vary the index \verb+strategy+ from $1$ to $3$. The index $1$
represents the estimation strategy for independent covariant POVMs, $2$ represents the AQSE strategy, $3$ represents the restricted-AQSE strategy. To reproduce the points for our proposed estimation scheme use the function
\begin{verbatim} 
ECI_Hvar(theta_real, par_space, n_boost, num_prob, n,a, C_lev, Margin_Err),
\end{verbatim}
where you have to set the confidence level \verb+C_lev+ and marginal error \verb+Margin_Err+ to the desired values. Finally, to calculate the curve for the entangled estimation strategy use the function
\begin{verbatim} 
Ent_Hvar(theta_real, par_space, n,a, num_prob).
\end{verbatim}

\bibliographystyle{plainnat}
\bibliography{librero1}

\begin{thebibliography}{35}
\providecommand{\natexlab}[1]{#1}
\providecommand{\url}[1]{\texttt{#1}}
\expandafter\ifx\csname urlstyle\endcsname\relax
  \providecommand{\doi}[1]{doi: #1}\else
  \providecommand{\doi}{doi: \begingroup \urlstyle{rm}\Url}\fi

\bibitem[Barndorff-Nielsen and Gill(2000)]{Barndorff-Nielsen2000}
O.~E. Barndorff-Nielsen and R.~D. Gill.
\newblock Fisher information in quantum statistics.
\newblock \emph{J. Phys. A}, 33\penalty0 (24):\penalty0 4481--4490, 2000.
\newblock \doi{10.1088/0305-4470/33/24/306}.

\bibitem[Berry and Wiseman(2001)]{Berry2001}
D.~W. Berry and H.~M. Wiseman.
\newblock {Adaptive measurements and optimal states for quantum
  interferometry}.
\newblock \emph{Technical Digest - Summaries of Papers Presented at the Quantum
  Electronics and Laser Science Conference}, \penalty0 (5):\penalty0 60--61,
  2001.
\newblock \doi{10.1109/QELS.2001.961853}.

\bibitem[Berry et~al.(2009)Berry, Higgins, Bartlett, Mitchell, Pryde, and
  Wiseman]{Berry2009}
D.~W. Berry, B.~L. Higgins, S.~D. Bartlett, M.~W. Mitchell, G.~J. Pryde, and
  H.~M. Wiseman.
\newblock {How to perform the most accurate possible phase measurements}.
\newblock \emph{Phys. Rev. A}, 80\penalty0 (5):\penalty0 1--22, 2009.
\newblock \doi{10.1103/PhysRevA.80.052114}.

\bibitem[Boixo and Somma(2008)]{Boxio}
S.~Boixo and R.~D. Somma.
\newblock Parameter estimation with mixed-state quantum computation.
\newblock \emph{Phys. Rev. A}, 77:\penalty0 052320, 2008.
\newblock \doi{10.1103/PhysRevA.77.052320}.

\bibitem[Braunstein and Caves(1994)]{Braunstein1994}
S.~L. Braunstein and C.~M. Caves.
\newblock Statistical distance and the geometry of quantum states.
\newblock \emph{Phys. Rev. Lett.}, 72:\penalty0 3439--3443, 1994.
\newblock \doi{10.1103/PhysRevLett.72.3439}.

\bibitem[Casella and Berger(2002)]{CaseBerg:01}
G.~Casella and R.L. Berger.
\newblock \emph{Statistical Inference}.
\newblock Duxbury advanced series in statistics and decision sciences. Thomson
  Learning, 2002.
\newblock ISBN 9780534243128.

\bibitem[Chapeau-Blondeau(2016)]{Chapeau2016}
F.~Chapeau-Blondeau.
\newblock Optimizing qubit phase estimation.
\newblock \emph{Phys. Rev. A}, 94:\penalty0 022334, 2016.
\newblock \doi{10.1103/PhysRevA.94.022334}.

\bibitem[Dee and Da~Silva(1998)]{Dee}
D.~Dee and A.~Da~Silva.
\newblock Maximum-likelihood estimation of forecast and observation error
  covariance parameters. part {I}: Methodology.
\newblock \emph{Mon. Weather Rev.}, 127, 09 1998.
\newblock \doi{10.1175/1520-0493(1999)127<1822:MLEOFA>2.0.CO;2}.

\bibitem[DeGroot and Schervish(2012)]{degroot}
M.~H. DeGroot and M.J. Schervish.
\newblock \emph{Probability and Statistics}.
\newblock Addison-Wesley, 2012.
\newblock ISBN 9780321500465.
\newblock \doi{10.1080/09332480.2013.845457}.

\bibitem[Dowling(2008)]{Jon}
J.~P. Dowling.
\newblock Quantum optical metrology – the lowdown on high-{N00N} states.
\newblock \emph{Contemporary Physics}, 49\penalty0 (2):\penalty0 125--143,
  2008.
\newblock \doi{10.1080/00107510802091298}.

\bibitem[Engle and McFadden(1994)]{Robert}
R.~Engle and D.~McFadden.
\newblock \emph{Handbook of Econometrics}, volume~4.
\newblock North Holland, 1994.
\newblock ISBN 0444887660,9780444887665.

\bibitem[Fr\"owis et~al.(2014)Fr\"owis, Skotiniotis, Kraus, and
  D\"ur]{Frowis2014}
F.~Fr\"owis, M.~Skotiniotis, B.~Kraus, and W.~D\"ur.
\newblock Optimal quantum states for frequency estimation.
\newblock \emph{New Journal of Physics}, 16\penalty0 (8):\penalty0 083010,
  2014.
\newblock \doi{10.1088/1367-2630/16/8/083010}.

\bibitem[Fujiwara(2006)]{Fujiwara2011}
A.~Fujiwara.
\newblock Strong consistency and asymptotic efficiency for adaptive quantum
  estimation problems.
\newblock \emph{J. Phys. A}, 39\penalty0 (40):\penalty0 12489--12504, sep 2006.
\newblock \doi{10.1088/0305-4470/39/40/014}.

\bibitem[Giovannetti et~al.(2004)Giovannetti, Lloyd, and Maccone]{Giovann2004}
V.~Giovannetti, S.~Lloyd, and L.~Maccone.
\newblock Quantum-enhanced measurements: Beating the standard quantum limit.
\newblock \emph{Science}, 306\penalty0 (5700):\penalty0 1330--1336, 2004.
\newblock \doi{10.1126/science.1104149}.

\bibitem[Helstrom(1969)]{Helstrom1969}
C.~W. Helstrom.
\newblock {Quantum detection and estimation theory}.
\newblock \emph{J. Stat. Phys.}, 1\penalty0 (2):\penalty0 231--252, 1969.
\newblock \doi{10.1007/BF01007479}.

\bibitem[Holevo(1982)]{Holevo2011}
A.~S. Holevo.
\newblock \emph{{Probabilistic and Statistical Aspects of Quantum Theory}}.
\newblock Springer Science \& Business Media, 1982.
\newblock \doi{10.1007/978-88-7642-378-9}.

\bibitem[Holevo(2005)]{Holevo2005}
A.~S. Holevo.
\newblock {Asymptotic estimation of a shift parameter of a quantum state}.
\newblock \emph{Theory of Probability and its Applications}, 49\penalty0
  (2):\penalty0 207--220, 2005.
\newblock \doi{10.1137/S0040585X97981044}.

\bibitem[Holevo(1973)]{Holevo1973}
A.S Holevo.
\newblock Statistical decision theory for quantum systems.
\newblock \emph{Journal of Multivariate Analysis}, 3\penalty0 (4):\penalty0 337
  -- 394, 1973.
\newblock \doi{10.1016/0047-259X(73)90028-6}.

\bibitem[Huang et~al.(2017)Huang, Motes, Anisimov, Dowling, and
  Berry]{Huang2017}
Z.~Huang, K.~R. Motes, P.~M. Anisimov, J.~P. Dowling, and D.~W. Berry.
\newblock Adaptive phase estimation with two-mode squeezed vacuum and parity
  measurement.
\newblock \emph{Phys. Rev. A}, 95:\penalty0 053837, 2017.
\newblock \doi{10.1103/PhysRevA.95.053837}.

\bibitem[Lehmann and Casella(1998)]{Lehmann1998}
E.~L. Lehmann and G.~Casella.
\newblock \emph{Theory of Point Estimation}.
\newblock Springer-Verlag, New York, NY, USA, second edition, 1998.

\bibitem[Manisera and Zuccolotto(2015)]{Marica}
M.~Manisera and P.~Zuccolotto.
\newblock {Identifiability of a model for discrete frequency distributions with
  a multidimensional parameter space}.
\newblock \emph{Journal of Multivariate Analysis}, 140:\penalty0 302--316,
  2015.
\newblock \doi{10.1016/j.jmva.2015.05.011}.

\bibitem[Martin et~al.(2020)Martin, Livingston, Hacohen-Gourgy, and
  Wiseman]{Martin2019}
L.~S. Martin, W.~P. Livingston, S.~Hacohen-Gourgy, and I.~Wiseman, H.
  M.and~Siddiqi.
\newblock Implementation of a canonical phase measurement with quantum
  feedback.
\newblock \emph{Nature Physics}, 2020.
\newblock \doi{10.1038/s41567-020-0939-0}.

\bibitem[Merzbacher(1998)]{Eugen}
E.~Merzbacher.
\newblock \emph{Quantum Mechanics}.
\newblock North Holland, 3 edition, 1998.
\newblock ISBN 0471887021, 9780471887027.

\bibitem[Monras(2006)]{PhysRevA.73.033821}
A.~Monras.
\newblock Optimal phase measurements with pure {Gaussian} states.
\newblock \emph{Phys. Rev. A}, 73:\penalty0 033821, 2006.
\newblock \doi{10.1103/PhysRevA.73.033821}.

\bibitem[Nagaoka(2005)]{Nagaoka}
H.~Nagaoka.
\newblock \emph{On the Parameter Estimation Problem for Quantum Statistical
  Models}, pages 125--132.
\newblock 2005.
\newblock \doi{10.1142/9789812563071_0011}.

\bibitem[Nielsen and Chuang(2010)]{Nielsen2000}
M.~A. Nielsen and I.~L. Chuang.
\newblock \emph{Quantum Computation and Quantum Information: 10th Anniversary
  Edition}.
\newblock Cambridge University Press, 2010.
\newblock \doi{10.1017/CBO9780511976667}.

\bibitem[Oh et~al.(2019)Oh, Lee, Rockstuhl, Jeong, Kim, Nha, and Lee]{Rock}
C.~Oh, C.~Lee, C.~Rockstuhl, H.~Jeong, J.~Kim, H.~Nha, and S.~Lee.
\newblock Optimal gaussian measurements for phase estimation in single-mode
  gaussian metrology.
\newblock \emph{npj Quantum Information}, 5\penalty0 (1):\penalty0 1--9, 2019.
\newblock \doi{10.1038/s41534-019-0124-4}.

\bibitem[Okamoto et~al.(2012)Okamoto, Iefuji, Oyama, Yamagata, Imai, Fujiwara,
  and Takeuchi]{Fuji2}
R.~Okamoto, M.~Iefuji, S.~Oyama, K.~Yamagata, H.~Imai, A.~Fujiwara, and
  S.~Takeuchi.
\newblock Experimental demonstration of adaptive quantum state estimation.
\newblock \emph{Phys. Rev. Lett.}, 109:\penalty0 130404, 2012.
\newblock \doi{10.1103/PhysRevLett.109.130404}.

\bibitem[Paris(2009)]{Paris2009}
M.~G.~A. Paris.
\newblock {Quantum estimation for quantum technology}.
\newblock \emph{International Journal of Quantum Information}, 7:\penalty0
  125--137, 2009.
\newblock \doi{10.1142/S0219749909004839}.

\bibitem[Peng and Fan(2020)]{Peng}
Y.~Peng and H.~Fan.
\newblock Feedback ansatz for adaptive-feedback quantum metrology training with
  machine learning.
\newblock \emph{Phys. Rev. A}, 101:\penalty0 022107, 2020.
\newblock \doi{10.1103/PhysRevA.101.022107}.

\bibitem[Pezz{\`{e}} and Smerzi(2014)]{Pezze2014}
L.~Pezz{\`{e}} and A.~Smerzi.
\newblock \emph{{Quantum theory of phase estimation}}, pages 691--741.
\newblock 2014.
\newblock \doi{10.3254/978-1-61499-448-0-691}.

\bibitem[Rodríguez-García(2020)]{Git_Repo}
M.~A. Rodríguez-García.
\newblock Qubit-phase-estimation.
\newblock \url{https://github.com/Gateishion/Quantum-Phase-Estimation.git},
  2020.

\bibitem[Toscano et~al.(2017)Toscano, Bastos, and de~Matos~Filho]{Toscano2017}
F.~Toscano, W.~P. Bastos, and R.~L. de~Matos~Filho.
\newblock Attainability of the quantum information bound in pure-state models.
\newblock \emph{Phys. Rev. A}, 95:\penalty0 042125, 2017.
\newblock \doi{10.1103/PhysRevA.95.042125}.

\bibitem[Wigner(1979)]{Wigner}
E.~P. Wigner.
\newblock \emph{{Symmetries and reflections}}.
\newblock Ox Bow Press., reprint edition edition, 1979.

\bibitem[Yamagata et~al.()Yamagata, Fujiwara, and
  Gill]{Yamagata2013AsymptoticQS}
K.~Yamagata, A.~Fujiwara, and R.~D. Gill.
\newblock Quantum local asymptotic normality based on a new quantum likelihood
  ratio.
\newblock \emph{Ann. Statist.}, \penalty0 (4):\penalty0 2197--2217, 08 .
\newblock \doi{10.1214/13-AOS1147}.

\end{thebibliography}

\end{document}